\documentclass[twoside]{article}
\newcommand{\mytitle}{Non-Archimedean Electrostatics}
\newcommand{\keywords}{non-archimedean analysis, thermodynamics, statistical physics, particle models, partition function, canonical ensemble, grand canonical ensemble, point processes, cylinder sets, local zeta functions}
\newcommand{\msc}{
60B20, 
60G55,  
82B23, 
11C08  
11R42  
11R04  
}

\usepackage{amsmath,amssymb,
  amsthm,amsfonts,amscd,mathrsfs,pifont,upgreek}
\usepackage{eucal}  
\usepackage{verbatim}
\usepackage{ifpdf}
\usepackage[colorlinks=true, citecolor=blue]{hyperref}

\ifpdf
        \usepackage[pdftex]{graphicx}
\else
        \usepackage[dvips]{graphicx}
\fi

\newtheorem{thm}{Theorem}[section]
\newtheorem{cor}[thm]{Corollary}

\newtheorem{lemma}[thm]{Lemma}

\newcommand{\qq}[1]{\qquad \mbox{#1} \qquad}

\newcommand{\BB}[1]{\ensuremath{\mathbb{#1}}}

\newcommand{\N}{\ensuremath{\BB{N}}}
\newcommand{\Q}{\ensuremath{\BB{Q}}}
\newcommand{\R}{\ensuremath{\BB{R}}}
\newcommand{\Z}{\ensuremath{\BB{Z}}}
\newcommand{\C}{\ensuremath{\BB{C}}}

\newcommand{\bs}{\ensuremath{\boldsymbol}}
\newcommand{\mf}{\ensuremath{\mathfrak}}
\newcommand{\wh}[1]{\widehat{#1}}
\newcommand{\rp}[1]{\mathrm{Re}(#1)}
\newcommand*\foo{\vcenter{\hbox{\includegraphics{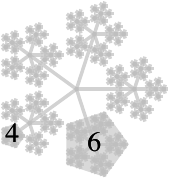}}}}
\usepackage[margin=1 in]{geometry}
\usepackage{mathtools}

\numberwithin{equation}{section}

\numberwithin{equation}{section}

\begin{document}
\title{\bfseries\sffamily \mytitle}
\author{{\sc Christopher D.~Sinclair}\footnote{Supported by Simons Foundation Collaboration Grant.}}
\maketitle

\begin{abstract}
We introduce ensembles of repelling charged particles restricted to a ball in a non-archimedean field (such as the $p$-adic rational numbers) with interaction energy between pairs of particles proportional to the logarithm of the ($p$-adic) distance between them. In the {\em canonical ensemble}, a system of $N$ particles is put in contact with a heat bath at fixed inverse temperature $\beta$ and energy is allowed to flow between the system and the heat bath. Using standard axioms of statistical physics, the relative density of states is given by the $\beta$ power of the ($p$-adic) absolute value of the Vandermonde determinant in the locations of the particles. The partition function is the normalizing constant (as a function of $\beta$) of this ensemble, and we identify a recursion that allows this to be computed explicitly in finite time. Probabilities of interest, including the probabilities that fixed subsets will have a prescribed number of particles, and the conditional distribution of particles within a subset given a prescribed occupation number, are given explicitly in terms of the partition function. We then turn to the {\em grand canonical ensemble} where both the energy and number of particles are variable. We compute similar probabilities to those in the canonical ensemble and show how these probabilities can be given in terms the canonical and grand canonical partition functions. Finally, we briefly consider the multi-component ensemble where particles are allowed to take different integer charges, and we connect basic properties of this ensemble to the canonical and grand canonical ensembles.
\end{abstract}

{\bf MSC2010:} \msc

\vspace{1cm}

\section{Introduction}

This paper lies at the intersection of number theory, probability and mathematical statistical physics.

We consider a collection of charged particles confined to a compact region of a complete non-archimedean field ({\em e.g.} $\Q_p$) at a fixed temperature. We may think of this as a non-archimedean ($p$-adic) {\em plasma}, and since the particles have identical charges, they have a tendency to repel. What might we want to know about this plasma? For starters we might want to know how many particles there are. This actually introduces two models: the {\em canonical ensemble} with a fixed number of particles, and the {\em grand canonical ensemble} where the number of particles is variable (but in a specific way suggested by physical laws). In either of these settings we might want to know the probability of finding a specified number of particles in a specified subregion of our domain. More specifically: Given a disjoint union of subregions, and an occupation number for each of these regions, what is the probability of each set having the specified number of particles? A specific, but particularly salient example, follows when we ask for the probability that a given subregion contains no particles. This latter probability is called a {\em gap probability} and for now we focus on this quantity as a proxy for more nuanced statistical information about counts of particles.

What might be a reasonable answer? Certainly a formula for this probability would be ideal. Moreover, if the subregion can be described using a finite amount of data, the ideal formula would require only a finite number of maneuvers to calculate this probability. Another possible solution would be to describe a recurrence for gap probabilities that terminate in a finite number of steps, given a finitely-described subregion. We will be aiming for the latter, and in the canonical ensemble, we will provide such recursions for gap probabilities and other common statistical quantities (like the free energy, partition function, etc).

While we believe our results are new, some of the results here appear in other guises in the literature. The  authors of \cite{MR2229382} investigate the probability that a polynomial with $p$-adic coefficients splits completely (has all roots) in the $p$-adic integers. The roots of such polynomials behave like our $p$-adic electrons at a very specific temperature. One of their main results gives a functional equation for the generating function (over degree of polynomials) of these probabilities. Here we report a similar functional equation for the grand canonical partition function, generalized to all temperatures (and with a new, different proof).

We point the reader to the recent preprint which gives certain important expectations in the canonical ensemble \cite{webster2020logcoulomb}.

In another direction, Igusa studied local zeta functions \cite{MR1743467} of which our canonical partition function (as a function of temperature) is a very specific examples. Examples of Igusa zeta functions similar to the canonical partition function appearing here can be found in \cite{MR1608309}.

Alternate titles for this paper include ``The $p$-adic Selberg Integral'' or ``$p$-adic Random Matrix Theory'' due to the appearance of a non-archimedean version of the Selberg Integral which appears as the partition function of the canonical ensemble. The Selberg Integral is an important special function \cite{MR0018287} and the fact that our partition function is a $p$-adic analog is reason enough to study it. See \cite{MR2434345} for a more comprehensive look at the importance of the Selberg integral. We also point out the recent preprint which considers more direct $p$-adic analogies of Selberg's integral \cite{fu2018selberg}.

The connection to random matrix theory is (currently) more tentative, since there are no random matrices introduced in this paper. However, in Hermitian random matrix theory, a Selberg-like integral appears as the normalization constant for certain ensembles of matrices \cite{MR2129906}. In some instances, determining a closed-form for those Selberg-like integrals leads to the solvability of the related ensemble of random matrices. The results we present here suggest that if this analogy holds for $p$-adic random matrices, then those ensembles are solvable in the sense that we can determine probabilities of interest about the locations and behaviors of the eigenvalues using the techniques outlined here.

The connection between random matrix theory and one and two-dimensional electrostatics is well-known, and indeed understanding the electrostatics provides insight into the eigenvalues of random matrices. This perspective was introduced by Dyson in a series of papers \cite{MR0143558} and explored in detail in the archimedean (real and complex) setting by Forrester in \cite{MR2641363}.

The current work is also connected to potential theory on non-archimedean spaces (See \cite{MR2482347} and the reference therein). Expressions similar to the potential energy of our particle system appear in that domain, where much work is done to investigate low-energy configurations especially as connected to problems in number theory \cite{MR3373827,MR3710764}. For low-temperatures we expect our particles to `jostle' around these low-energy configurations, and it would be worthwhile to explore the implications of our results to fluctuations about the ground state as it arises in potential theory.

For a broader survey of $p$-adic mathematical physics, see \cite{MR3652559}. $p$-adic integrals related to our canonical partition function can be found in \cite{MR1608309}.

We provide a brief introduction to both non-archimedean fields and statistical physics as we go. A more complete (and well-written) introduction to $p$-adic numbers can be found here \cite{MR1488696}. Likewise an approachable introduction to statistical physics can be found in David Tong's lecture notes on the subject \cite{tong}.

\section{Non-Archimedean Fields}

\subsection{Absolute Values}

An absolute value $| \cdot |$ on $\Q$ satisfies the axioms:
\begin{enumerate}
\item $| x | \in [0, \infty)$ with $|x| = 0$ iff $x=0$;
\item $|x y| = |x| |y|$;
\item \label{triangle inequality} $|x + y| \leq |x| + |y|$.
\end{enumerate}
The absolute value with $|0|=0$ and $|x|=1$ for all other $x$ is called the {\em trivial} absolute value and we will exclude it from all consideration. The usual absolute value $|x|_{\infty} := \mathrm{sign}(x) x$ is, of course, an absolute value. Given a prime integer $p$, we may factor $x \in \Q$ as
\[
x = p^b \frac{m}{n}
\]
where $b,m,n$ are integers with $m$ and $n$ relatively prime to $p$. The $p${\em-adic absolute value} is then specified by
\[
\left|p^b \frac{m}{n}\right|_p = p^{-b}.
\]
It is easily varified that $|\cdot|_p$ is an absolute value, which satisfies a stronger version of \ref{triangle inequality}, called the {\em strong triangle inequality}
\begin{equation}
  |x + y|_p \leq \max\{|x|_p, |y|_p\}.
\end{equation}
Absolute values which satisfy the strong triangle inequality are called {\em non-archimedean} absolute values.

Two absolute values are equivalent if one is a power of the other, and an equivalence class of absolute values is called a {\em place} of $\Q$. A celebrated theorem of Ostrowski \cite{Ostrowski1916} shows that any non-trivial absolute value on $\Q$ is equivalent to either the usual absolute value $|\cdot|_{\infty}$ or to $|\cdot|_p$ for some prime $p$.

\subsection{Completions}

The real numbers are constructed from the rational numbers by completing $\Q$ with respect to $|\cdot|_{\infty}$.  Recall the construction: two Cauchy sequences of rational numbers $(x_n)$ and $(y_n)$ are equivalent if $(x_n - y_n)$ converges to 0. The real numbers are then defined to be the set of equivalence classes of Cauchy sequences, and the algebraic operations of addition and multiplication are given by coordinate-wise addition and multiplication of equivalence class representatives.  The rational numbers can be represented by constant sequences, and these are dense in the completion, and using this fact we may extend the absolute value $|\cdot|_{\infty}$ to the usual absolute value on $\R$.

The $p$-adic numbers $\Q_p$ are constructed in the same way, except that the notions of convergence are with respect to the $p$-adic absolute value.  That is, $(x_n)$ is Cauchy if given $\epsilon > 0$ there exists $M$ such
\[
\sup_{n,m > M} |x_n - x_m|_p < \epsilon
\]
and we say $(x_n)$ is equivalent to $(y_n)$ if there exists $M$ such that
\[
\sup_{m > M} |x_m - y_m|_p < \epsilon.
\]
As with the real numbers, $\Q$ is dense in $\Q_p$ and the absolute value $|\cdot|_p$ extends to an absolute value on $\Q_p$.

\subsection{Differences and Similarities Between $\R$ and $\Q_p$}

Despite the similarity in construction, $\Q_p$ has important structural differences from $\R$. Here is a brief compendium of facts about $\Q_p$ which illustrate such differences:
\begin{enumerate}
\item $|\cdot|_p$ takes values in the discrete set $\{p^{n} : n \in \Z\}$.
\item If $n \in \Z$, then $|n|_p \leq 1$.
\item The completion of $\Z$ in $\Q_p$ is given by $\Z_p := \{x \in \Q_p : |x| \leq 1\}$ and called the $p$-adic integers. $\Z$ is dense in $\Z_p$.
\item $\Z_p$ is a ring with a unique maximal ideal ${\mf m}_p := \{x \in \Z_p : |x| < 1\}$. Moreover ${\mf m}_p$ is a principal ideal generated by $p$ (i.e.~$\mf m_{\mf p} = p\Z_p$).
\item $\Z_p / {\mf m}_p \cong \Z/p\Z$.
\end{enumerate}

In spite of these differences, there are also similarities which we will exploit.  Perhaps most important is that both $\R$ and $\Q_p$ are locally compact abelian groups under addition, and $\R^{\times}$ and $\Q_p^{\times}$ are locally compact abelian groups under multiplication. This is useful because it means that $\Q_p$, like $\R$, has a Haar measure.

\subsection{Haar Measure on $\Q_p$}

A Haar measure on a locally compact abelian group is a Borel measure which is invariant under the action of the group on itself. For instance, Lebesgue measure is a Haar measure on $\R$, since the Lebesgue measure of an interval (and hence any Lebesgue measurable set) is invariant under translation. Haar measures are not unique, though once one specifies the measure of a compact set containing an open set, the measure is completely specified.  Thus we may contruct a unique Borel measure $\mu_p$ on $\Q_p$ with the following properties: $\mu_p(\Z_p) = 1$ and for any $x \in \Q_p$ and Borel subset $B$, $\mu_p(x + B) = \mu_p(B)$. This measure also behaves nicely with respect to multiplication $\mu_p(x B) = |x|_p \mu_p(B)$. In particular, $\mu_p({\mf m}_p) = \mu_p(p \Z_p) = 1/p$. We remark that $\Z_p$ is a compact abelian group, and $\mu_p$ restricted to $\Z_p$ is the unique Haar probability measure on this group.

We will have limited need for a Haar measure on $\Q_p^{\times}$, but for the record, it is absolutely continuous with respect to $\mu_p$ on $\Q_p^{\times}$, and a natural Haar measure is given by $\mu_p^{\times}(dx) = \mu_p(dx)/|x|_p$.

\subsection{Non-Archimedean Completions of Number Fields}

We may generalize the previous discussion somewhat by letting $K$ be a number field with ring of integers ${\mf o}$ and a chosen prime ideal $\mf p \subset \mf o$. $\mf o/\mf p$ is a finite field, say $\mathbb{F}_q$ where $q$ is a power of a rational prime.  Each element $x$ of $K$ lies in $\mf p^n/\mf p^{n+1}$ for some rational integer $n$, and we define $|x|_{\mf p} = 1/q^n$. Completing $K$ with respect to $|\cdot|_{\mf p}$ produces the field $K_{\mf p}$. As before, we define the ring of integers of $K_{\mf p}$ and its unique maximal ideal by
\[
\mf o_{\mf p} = \{x \in K_{\mf p} : |x|_{\mf p} \leq 1\} \qq{and}
\mf m_{\mf p} = \{x \in \mf o_{\mf p} : |x|_{\mf p} < 1 \}.
\]
The units in $\mf o_{\mf p}$ are given by $U_{\mf p} = \{x \in \mf o_{\mf p} :  |x|_{\mf p} = 1 \} = \mf o_{\mf p} \setminus \mf m_{\mf p}$.

By general principles, $\mf o_{\mf p}/\mf m_{\mf p}$ is a finite field which can be shown to be isomorphic to $\mathbb{F}_q$.  That is, $\mf m_{\mf p}$ has $q$ cosets, and we denote these by $\mf m_{\mf p}, 1 + \mf m_{\mf p}, \ldots, q - 1 + \mf m_{\mf p}$. It can be shown that $\mf m_{\mf p}$ is a principal ideal, and if $\pi$ is a generator (or {\em uniformizer}) for $\mf m_{\mf p}$ then $|\pi|_{\mf p} = 1/q$.

There is a unique Haar measure $\mu_{\mf p}$ on $\mf o_{\mf p}$ satisfying $\mu_{\mf p}(\mf o_{\mf p}) = 1$. Since $\mf o_{\mf p}$ is the disjoint union of $q$ cosets of $\mf m_{\mf p}$ each of which is a translation of $\mf m_{\mf p}$, we have that $\mu_{\mf p}(\mf m_{\mf p}) = 1/q$.

\begin{figure}[h!]
\centering
\includegraphics[scale=.4]{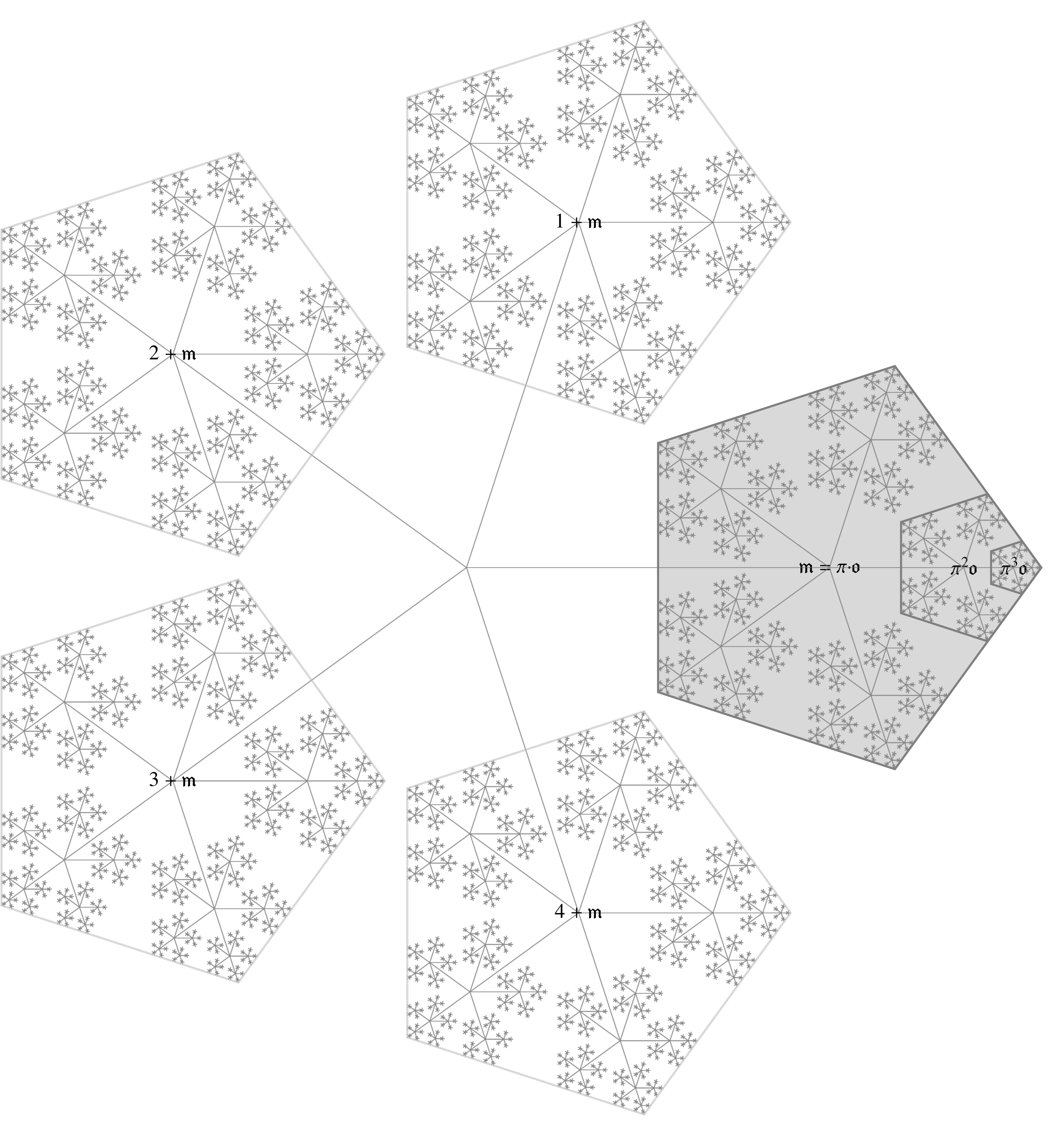}
\caption{A schematic diagram of $\mf o = \mathbb Z_5$. Cosets and powers of the maximal ideal $\mf m$ are represented. The $5$-adic numbers are the fractal boundary of the tree, {\em not} the interstitial edges and vertices. Every ball in $\mathbb Z_5$ is the fractal boundary of one of the naturally appearing `pentagons' in the tree (at all scales, assuming it had infinite resolution), and each is naturally in bijective equivalence with $\mathbb Z_5$ itself. Addition by a fixed $\alpha$ in $\mathbb Z_5$ can be thought of as an epicyclic rotation where the cosets of $\mf m$ are rotated, as are their neighborhoods of the next smallest radius, and so on down the tree. Addition by a fixed rational integer can be distinguished from addition by an arbitrary $\alpha$ since in the former situation there is some radius below which the epicyclic rotations are trivial. Multiplication by a unit (any element not in $\mf m$) stabilizes $\mf m$ and permutes in some manner the non-trivial cosets of $\mf m$. Multiplication by an element of $\mf m$ sends $\mf o$ onto some smaller ideal $\pi^r \mf o$; a contraction of sorts onto a smaller ball. There are many ways to embed $\mathbb Z$ into this diagram, but one simple model has $0$ as the right-most point on the diagram, and given any rational integer there is some `pentagon' for which that integer is its right-most point. This provides a heuristic verification of the density of the rational integers in $\mathbb Z_5$. Haar measure corresponds to areal measure: the largest `pentagon' (all of $\mathbb Z_5$) is assumed to have measure/area 1, and then the measure of any smaller `pentagon' is its area. }
\end{figure}

\subsection{Notation}

We will be working in a single non-archimedean field of characteristic 0. We may take this to be the completion of a number field with respect to an absolute value induced by a prime ideal, but we need not burden our notation with explicit dependence on the particular number field or the prime ideal.

Thus, we set $\mathbb K$ to be a field complete with respect to non-archimedean absolute value $|\cdot|$, with ring of integers $\mf o$, maximal ideal $\mf m$ with $q$ (a prime power) cosets generated by uniformizer $\pi$, and Haar measure $\mu$ normalized so that $\mu(\mf m) = 1/q$. We will often be integrating over the cartesian product of a number of copies of a subset of $\mf o$. We will denote the product measure on $\mf o^N$ by $\mu^N$ (this is the unique Haar probability measure on $\mf o^N$). There is ambiguity interpreting $\mf m^N$ and we will interpret it as the $N$-fold copy of the maximal ideal $\mf m$. If we need to denote the {\em ideal} given by the $n$th power of $\mf m$--all such ideals in $\mf o$ are of this form--we will write $\pi^n \mf o$.

\section{Electrostatics}

Imagine two like charged particles identified with points $\alpha$ and $\alpha'$ in $\mf o$.  We define the {\em interaction energy} of this simple system by
\[
E(\alpha, \alpha') := -\log|\alpha - \alpha'|.
\]
Note that $E(\alpha, \alpha') \in [0,\infty]$, with $E(\alpha) = \infty$ if and only if $\alpha = \alpha'$.  Notice also that $E(\alpha, \alpha')$ takes its minimal value 0, exactly when $(\alpha - \alpha')$ is a unit---that is, when $\alpha$ and $\alpha'$ are in different cosets of $\mf m$.

Given a system of $N$ such particles, its potential energy is the sum of interaction energies over all $N \choose 2$ pairs of particles.  That is, if we identify the {\em state} of a system with $N$ charged particles by $\bs \alpha \in \mf o^N$, the potential energy of that state is given by
\begin{equation}
\label{eq:1}
E(\bs \alpha) = -\sum_{m < n} \log|\alpha_n - \alpha_m|.
\end{equation}
As defined, each physical state is overcounted by a factor of $N!$ since permuting the coordinates of $\bs \alpha$ does not alter the identity of the system.\footnote{We are implicitly deciding that the particles are indistinguishable, and the probability that two particles are co-located is zero---the first of these assumptions is definitional and the second will be justified in later sections.}  This overcounting will be adjusted for later.

\subsection{The Microcanonical Ensemble}

An {\em ensemble} is a probability measure on a set of states of a physical system.

The {\em microcanonical} ensemble is that given by our system conditioned so that the total energy of the system is some fixed value $E_{\ast}$.  Since our  absolute value is discrete, not all values are allowed for $E_{\ast}$, and we will assume that $E_{\ast}$ is an attainable value of the energy as specified by \eqref{eq:1}.  The set of attainable states of the microcanonical ensemble is then
\[
\Omega_N(E_{\ast}) := \{ \bs \alpha \in \mf o^N : E(\bs \alpha) = E_{\ast} \}.
\]
A first obvious question is what is the volume of the set of attainable states?  That is, what is $\mu^N(\Omega_N(E_{\ast}))$?

To answer this question, it will be convenient to define
\begin{equation}
\label{eq:2}
F_N(\xi) := \mu^N\{ \bs \alpha \in \mf o^N : E(\bs \alpha) \leq \log \xi\} \qquad \xi \geq 0.
\end{equation}
It follows then that the volume of accessible states for the energy $E_{\ast}$ is
\[
\mu^N(\Omega_N(E_{\ast})) = F_N(e^{E_\ast}) - \lim_{E \rightarrow E_{\ast}-} F_N(e^E).
\]
This volume is 0 if $E_{\ast}$ is not an allowable value of the energy.

The introduction of the $\log \xi$ term on the right-hand-side of \eqref{eq:2} will be convenient in the sequel, since we can rewrite $F_N(\xi)$ as
\[
F_N(\xi) = \mu^N\bigg\{ \bs \alpha \in \mf o^N : |\Delta_N(\bs \alpha)| < \xi \bigg\},
\]
where $\Delta_N$ is the Vandermonde determinant
\[
\Delta_N(\bs \alpha) := \prod_{m<n} (\alpha_n - \alpha_m).
\]
Written in this way, we see that $F_N$ is the cumulative distribution function for the random variable $|\Delta_N(\bs \alpha)|$ where the $\alpha_1, \ldots, \alpha_N$ are independent uniform random variables in $\mf o$.

\section{The Canonical Ensemble}

The canonical ensemble differs from the microcanonical ensemble in that we now allow the energy to vary.  Since energy is a conserved quantity, we do this by placing our system of particles in contact with another, typically much larger system, so that the energy of the aggregate system is constant, but energy is allowed to flow between our system of $N$ particles and the larger system.  The larger system is called a {\em heat reservoir} and we will view it as being at a fixed temperature $T$.\footnote{While temperature has a precise definition in thermodynamics, we need not dwell on it here. For the purposes of this paper, the temperature should be thought of as a parameter which controls how random the particles in the system are; when $T=0$ the particles are `frozen' in a low-energy configuration, when $T = \infty$ the thermal fluctuations overwhelm any electrostatic effects.}  It is usually more convenient to introduce the inverse temperature parameter $\beta = (kT)^{-1}$.  Here $k$ is Boltzmann's constant and $\beta$ is a unitless quantity.

The relative density of states is given by the {\em Boltzmann factor} $e^{-\beta E(\bs \alpha)}$.  That is, the probability density of finding the system in state $\bs \alpha$ is given by
\begin{equation}
\label{eq:3}
\frac{1}{Z} e^{-\beta E(\bs \alpha)} = \frac{1}{Z}\prod_{m < n} |\alpha_n - \alpha_m|^{\beta}
\end{equation}
where
\[
Z = \int_{\mf o^N} \prod_{m<n} |\alpha_n - \alpha_m|^{\beta} \, d\mu^N(\bs \alpha).
\]
$Z$ is called the {\em partition} function of the {\em canonical ensemble} of particles, and it is more than just a normalization constant necessary to make a probability measure.

It is sometimes useful to make explicit the variables on which $Z$ depends, and so we write
\[
Z(N, V, \beta) = \int_{V^N} |\Delta_N(\bs \alpha)|^{\beta} \, d\mu^N(\bs \alpha).
\]
Here $V$ is a finite measure subset of $K_{\mf p}$, whose measure plays the role of volume in traditional statistical physics.  For our purposes, $V$ will usually be either $\mf o$ or $\mf m$. We define $Z(0, V, \beta) := 1$ and, with the interpretation that an empty product is equal to 1, we see that $Z(1, V, \beta) = \mu(V)$. In particular, $Z(0, \mf o, \beta) = Z(1, \mf o, \beta) = 1$.

For those uncomfortable with the sudden introduction of temperature, or those unfamiliar with the derivation of the Boltzmann factor, we can take \eqref{eq:3} as an axiomatic relationship between the energy and density of states, and the temperature as represented by $\beta$.  Let us see, however that for the value $\beta = 0$, the density of states satisfies our intuition:  When $\beta = 0$ the temperature is infinite.  It is reasonable to suppose that in such a situation thermal fluctuations of particle positions will overwelm any repulsion stemming from electrical charge.  That is, when $\beta = 0$ the particles are independent and uniform over $\mf o$.  That is the relative density of states should be constant on $\mf o^N$.  The intuition in this case agrees with the result given by \eqref{eq:3}.

Similarly, as $\beta \rightarrow \infty$, the temperature is tending toward 0, and we expect the repulsion from the charge to overwhelm the thermal fluctuations. From a physical perspective, we therefore expect that the system will find itself in a state with minimal energy.  The minimal energy configurations correspond to states $\bs \alpha$ where $|\Delta(\bs \alpha)|$ is maximal (note that the maximum is attained because $|\Delta(\bs \alpha)| \leq 1$ and the absolute value is discrete).  Looking at the integrand in $Z(N,V,\beta)$ as $\beta \rightarrow \infty$, the contributions to the integral from low energy configurations (exponentially!) overwhelm higher energy configurations, and the resulting density of states becomes localized around the states with minimal energy.

\subsection{Relating the Microcanonical and Canonical Ensembles}

Our goal in the study of the microcanonical ensemble is to determine $\mu^N(\Omega_N(E_{\ast}))$.  This information is encoded into the distribution function $F_N$ and we will attempt to derive useful information about $F_N$ by considering its Mellin transform\footnote{For those unfamiliar with the Mellin transform, it is a multiplicative version of the Fourier transform and has the same utility.  In particular, there is an inversion formula which links a function and its Mellin transform.}
\[
\wh{F_N}(s) := \int_0^{\infty} \xi^{-s} F_N(\xi) \,\frac{d\mu(\xi)}{\xi}.
\]
The following lemma relates the Mellin transform of $F_N$ to $Z(N, \mf o, \beta)$.
\begin{lemma}
  \label{mellin lemma}
For $\rp{s} > 0$,
$\wh{F_N}(s) = -\frac{1}{s} Z(N, \mf o, s)$.
\end{lemma}
\begin{proof}
Using Lebesgue-Stieltjes integration by parts,
\begin{align*}
\int_0^{\infty} \xi^{s} F_N(\xi) \,\frac{d\mu(\xi)}{\xi} &= \left. \frac{\xi^{s}}s F_N(\xi) \right|_0^{\infty} - \frac{1}s \int_0^{\infty} \xi^{s} dF_N(\xi).
\end{align*}
Since $F_N(0) = 0$ and $F_N(\xi) \leq 1$, the first term vanishes, and we find
\[
\wh{F_N}(s) = -\frac{1}s \int_0^{\infty} \xi^{s} dF_N(\xi) = -\frac1s \int_{\mf o^N} |\Delta(\bs \alpha)|^{s} \, d\mu^N(\bs \alpha)
\]
as claimed.
\end{proof}

\subsection{The Additivity of Energy over Cosets of $\mf m$}

By our previous remarks, if $\alpha$ and $\alpha'$ are particles in different cosets of $\mf m$, then this pair of particles does not contribute to the energy of the configuration. This is a primary observation: {\it Particles in different cosets can't `sense' each other}.
\begin{lemma}[Additivity of energy over cosets.]
Suppose $\bs \alpha$ is a state with $n^0$ particles in $\mf m$, $n^1$ particles in the coset $1 + \mf m$, etc. If we represent the state in the $r$th coset by $\bs \alpha^r$, then by reordering the coordinates if necessary, we can write $\bs \alpha = (\bs \alpha^0, \bs \alpha^1, \ldots, \bs \alpha^{q-1})$. Then,
\[
E(\bs \alpha) = \sum_{r=0}^{q-1} E(\bs \alpha^r).
\]
Alternatively,
\[
|\Delta(\bs \alpha)| = \prod_{r=0}^{q-1} |\Delta(\bs \alpha^r)|.
\]
\end{lemma}
We will call $(\bs \alpha^0, \bs \alpha^1, \ldots, \bs \alpha^{q-1})$ a {\em factored state} of $\bs \alpha$ with occupancy vector $\mathbf n = (n^0, \ldots, n^{q-1})$ which sums to $N$. The set of factored states with occupancy vector $\mathbf n$ is given by
\[
U_{\mathbf n} = \mf m^{n^0} \times (1 + \mf m)^{n^1} \times \cdots \times (q - 1 + \mf m)^{n^{q-1}}
\]

\subsection{The Partition Function in the Canonical Ensemble}

Here we derive a way of expressing $Z(N, \mf o, \beta)$ in terms of $Z(n, \mf m, \beta)$ for $n < N$. This will provide a recursive way to determine $Z(N, \mf o, \beta)$.
\begin{thm}
  \label{recursionthm}
  For $N > 0$,
\[
  Z(N,\mf o, \beta)
  = N! \sum_{\mathbf n} \bigg\{\prod_{r=0}^{q-1} \frac{Z(n^r, \mf m, \beta)}{n^r!}
  \bigg\},
\]
where the sum is over all occupancy vectors $\mathbf n$ with $n^0 + \cdots + n^{q-1} = N$.
\end{thm}
\begin{proof}
If $\alpha = \alpha'$ for some two coordinates of $\bs \alpha$, then $\Delta(\bs \alpha) = 0$. This means that such states make no contribution to $Z(N, \beta, \mf o)$ and we can safely ignore such inadmissable states.
Each admissible state with occupancy vector $\mathbf n$ corresponds to $n_0! \cdots n_{q-1}!$ factored states which differ by permuting the particles in each coset seperately. Moreover, each admissible factored state corresponds to $N!$ admissible states formed by permuting the coordinates indiscriminantly. Thus,
\[
Z(N,\mf o, \beta) = N! \sum_{\mathbf n} \bigg\{\prod_{r=0}^{q-1} \frac{1}{n_r!}\bigg\} \int_{U_{\mathbf n}} |\Delta(\bs \alpha_0, \ldots, \bs \alpha_{q-1} )|^{\beta} \, d\mu^{n_0}(\bs \alpha_0) \cdots d\mu^{n_{q-1}}(\bs \alpha_{q-1}).
\]
The additivity of energy over cosets implies that the integrand factors, and Fubini's theorem implies then that
\[
Z(N,\mf o, \beta) = N! \sum_{\mathbf n} \bigg\{\prod_{r=0}^{q-1} \frac{1}{n_r!} \int_{(r + \mf m)^{n_r}} |\Delta(\bs \alpha_{r} )|^{\beta} \, d\mu^{n_{r}}(\bs \alpha_{r})\bigg\}.
\]
As a final maneuver, we note that both the integrand and the measure are invariant under the change of variables $\bs \alpha_{r} \mapsto r + \bs \alpha_{r}$.  That is, the physics can't distinguish the identity of cosets so we can replace the integrals over individual cosets with integrals over independent copies of $\mf m$. In any event,
\[
Z(N,\mf o, \beta) = N! \sum_{\mathbf n} \bigg\{\prod_{r=0}^{q-1} \frac{1}{n_r!} \int_{\mf m^{n_r}} |\Delta(\bs \alpha)|^{\beta} \, d\mu^{n_{r}}(\bs \alpha)\bigg\}
= N! \sum_{\mathbf n} \bigg\{\prod_{r=0}^{q-1} \frac{Z(n_r, \mf m, \beta)}{n_r!}
\bigg\}. \qedhere
\]
\end{proof}

By rescaling the domain we can write $Z(n, \mf m, \beta)$ in term of $Z(n, \mf o, \beta)$ which leads to a recursive formula for $Z(N, \mf o, \beta)$.
\begin{lemma}
  \label{lemma:1}
If $B = \zeta + \pi^r \mf o$ then,
$
Z(n, B, \beta) = q^{-r\beta {n \choose 2} - rn} Z(n, \mf o, \beta).
$
\end{lemma}
\begin{proof}
This is a special case of Lemma~\ref{push down lemma} below.
\end{proof}

This lemma leads immediately to the following theorem.
\begin{thm}
  \label{recursion2thm}
For $N > 0$,
\begin{equation}
\label{eq:zn}
  Z(N, \mf o, \beta) = N! \sum_{\mathbf n}\bigg\{\prod_{r=0}^{q-1} \frac{q^{-\beta{n_r \choose 2} -n_r}}{n_r!} Z(n_r, \mf o, \beta)
  \bigg\}.
\end{equation}
Solving for $Z(N, \mf o, \beta)$ (which appears on both sides of \ref{eq:zn}),
\[
Z(N, \mf o, \beta) =
\left( q^N -  q^{1-\beta{N \choose 2}} \right)^{-1} N! \sum_{\mathbf n}{\big.}' \bigg\{ \prod_{j=0}^{q-1} \frac{q^{-\beta{n_r \choose 2}}}{n_r!} Z(n_j, \mf o, \beta)
\bigg\},
\]
where $\sum'_{\mathbf n}$ is over all occupancy vectors except those of the form $(0, \ldots, N, \ldots 0)$ (that is, except those that correspond to all particles being in the same coset of $\mf m$).
\end{thm}

\begin{thm}
  \label{quad-rec}
The $\{ Z(n, \mf o, \beta): n=0,1,\ldots,N-1 \}$ satisfy
\[
\sum_{n=0}^N \frac{(N - (q+1)n)}{n!(N-n)!} q^{-{n \choose 2
} \beta} q^{-n} Z(n, \mf o, \beta) Z(N-n, \mf o, \beta) = 0.
\]
\end{thm}
See the remark after Theorem~\ref{gcz} for the proof.

Theorem~\ref{quad-rec} gives us an easy way to compute $Z(N, \mf o, \beta)$ for small values of $N$. These are increasingly complicated rational functions in $q^{-\beta}$.
\begin{align*}
  Z(0, \mf o, \beta) & = 1 \\
  Z(1, \mf o, \beta) & = 1 \\
  Z(2, \mf o, \beta) & = \frac{(q-1) q^{\beta }}{q^{\beta +1}-1} \\
  Z(3, \mf o, \beta) & = \frac{(q-1) q^{3 \beta } \left(-2 q^{\beta +1}+q^{\beta +2}+2 q-1\right)}{\left(q^{\beta
   +1}-1\right) \left(q^{3 \beta +2}-1\right)} \\
  Z(4, \mf o, \beta) & = \frac{-\frac{(q-1)^2 (4-2 (q+1)) q^{\beta -2}}{4 \left(q^{\beta
   +1}-1\right)^2}-\frac{(3-q) (q-1) \left(-2 q^{\beta +1}+q^{\beta +2}+2 q-1\right) q^{3
   \beta -1}}{6 \left(q^{\beta +1}-1\right) \left(q^{3 \beta +2}-1\right)}-\frac{(q-1)
   (4-3 (q+1)) \left(-2 q^{\beta +1}+q^{\beta +2}+2 q-1\right)}{6 q^3 \left(q^{\beta
   +1}-1\right) \left(q^{3 \beta +2}-1\right)}}{\frac{1}{24} (4-4 (q+1)) q^{-6 \beta
   -4}+\frac{1}{6}}.
\end{align*}

\subsection{The distribution of energies in the microcanonical ensemble}

The observation that $Z(N, \mf o, \beta)$ is expressible in terms of the Mellin transform of $F_N$ (Lemma~\ref{mellin lemma}) means that analytic information, viewing $Z(N, \mf o, \beta)$ as a function of a complex variable $\beta$, can provide information about the nature of the distribution of allowable energies in the microcanonical ensemble. In order to distinguish $\beta$ the complex variable in this expression from the inverse temperature (which must necessarily be real and positive) we will write $s = \sigma + i t$ for a complex variable and write $Z(N, \mf o, s)$ for the partition function as a function of a complex variable.

A first observation is that $Z(N, \mf o, s)$ is analytic in the right half plane $\{\sigma + it : \sigma > 0\}$. To see this suppose $\sigma > 0$ and with the usual absolute value on $\C$ denoted by $|\cdot|_{\infty}$,
\begin{equation}
  \label{Zbound}
\int_{\mf o^N} \left| |\Delta(\bs \alpha)|^{\sigma + it} \right|_{\infty} \, d\mu^N(\bs \alpha) = \int_{\mf o^N} |\Delta(\bs \alpha)|^{\sigma} \, d\mu^N(\bs \alpha) \leq \int_{\mf o^N} \, d\mu^N(\bs \alpha) = 1.
\end{equation}
It follows that, if $T$ is an oriented triangle in $\{ \sigma > 0 \}$ then
\[
\int_T Z(N, \mf o, s) \, ds = \int_{\mf o^N} \bigg\{\int_T |\Delta(\bs \alpha)|^{s} ds \bigg\} \, d\mu^N(\bs \alpha),
\]
where we used \eqref{Zbound} to justify the use of Fubini's Theorem, and the $ds$ integrals are (complex) line integrals around $T$. But, for  every fixed $\bs \alpha$, the function $s \mapsto |\Delta(\bs \alpha)|^{s}$ is analytic and hence
\[
\int_T |\Delta(\bs \alpha)|^{s} ds = 0 \qq{and thus} \int_T Z(N, \mf o, s) \, ds = 0.
\]
Since this is true for all triangles in the right half-plane, Morera's Theorem implies $Z(N, \mf o, s)$ is analytic there. In fact, we will see below that the domain of convergence of $Z(N, \mf o, s)$ is the half-plane $\sigma > -2/N$.
\begin{thm}
\label{abscissa of convergence} The integral defining $Z(N, \mf o, s)$ converges to an analytic function of $s = \sigma + it$ in the half-plane $\{ \sigma > -2/N\}$ and is absolutely divergent on $\{\sigma \leq -2/N\}$.
\end{thm}
The proof of this theorem will come after the development of a handful of lemmas.

One importance of recognizing the Mellin transform of $F_N(\xi)$ in terms of
$Z(N, \mf o, s)$ is that information from any analytic/meromorphic  continuation
of $Z(N, \mf o, s)$ beyond the initial domain of convergence gives us new
information about $F_N(\xi)$.

Our first observation is that $Z(N, \mf o, s)$ analytically continues to a meromorphic function. In fact it is a rational function in $q^{-s}$. We will provide a proof of this fact, but it is also the consequence of a much deeper theorem of Igusa on the continuation of certain types of $L$-functions, of which $Z(N, \mf o, s)$ is a particular example.
\begin{lemma}
For $N \geq 2$, there exists a non-constant rational function $R_N$ with rational coefficients so that $Z(N, \mf o, s) = R_N(q^{-s})$ for all $s = \sigma + it$ with $\sigma > 0$.
\end{lemma}
\begin{proof}
The proof is an easy consequence of Theorem~\ref{recursion2thm}. By definition $Z(0, \mf o, s) = 1$ and a trivial calculation shows $Z(1, \mf o, s) = 1$ as well. The second equation in Theorem~\ref{recursion2thm} then implies that $Z(N, \mf o, s)$ can be expressed as a ratio where the numerator is a rational linear combination of products of $Z(n, \mf o, s)$ with $n < N$,  and the denominator is a polynomial in $q^{-s}$ with rational coefficients. The strong inductive hypothesis is that $Z(n, \mf o, s)$ is a rational function in $q^{-s}$ with rational coefficients for all  $2 \leq n < N$, and the result follows.
\end{proof}

\subsection{The $\sigma$-algebra of symmetrized sets}

An ensemble is merely a probability space, and we here we set up the formal  machinery to compute probabilities of events of physical interest.

We set $\mathcal B$ and $\mathcal B^N$ to be the Borel $\sigma$-algebras on $\mf o$ and $\mf o^N$ as usual. We think of $\mf o$ as the one-particle space and $\mf o^N$ as the state space of our system of $N$ particles. The $\sigma$-algebra $\mathcal B^N$, however, is too large for our purposes in the sense that it contains events that would be out of reach of an observer of the system. Consider, for instance the set $A \times \mf o \times \cdots \times \mf o \subseteq \mf o^N$ for some $A \in \mathcal B$. This event is equivalent to knowing whether or not the particle in the first coordinate of the vector $\bs \alpha$ is in $A$.
However, since our particles are indistinguishable, we can't discern whether or not the {\em first} particle is in $A$; the closest we can come is to discern whether or not {\em one} of the particles is in $A$.

This reasoning suggests we should only consider events in $\mf o^N$ which are stabilized by the natural action of the symmetric group $S_N$. We denote the $\sigma$-algebra generated by all such {\em symmetrized} Borel subsets by $\mathcal S_N$. We explicitly define our probability space by $(\mf o^N, \mathcal S_N, \mathbb P_N)$ where,
\[
\mathbb P_N(B) = \frac{1}{Z(N,\beta,\mf o)} \int_B \prod_{m < n} |\alpha_n - \alpha_m|^{\beta} \, d\mu^N(\bs \alpha).
\]
$\mathbb P_N(B)$ depends implicitly on $\beta$. If we need to make this dependence explicit, we will write $\mathbb P(N, B, \beta)$.

To see what events in $\mathcal S_N$ look like, consider a rectangle $A = A_1 \times A_2 \times \cdots \times A_N$ in $\mathcal B^N$. Given a permutation $\tau \in S_N$ we define $\tau \cdot A = A_{\tau(1)} \times A_{\tau(2)} \times \cdots \times A_{\tau(N)}$. Using this, we define
\[
S_N \cdot A = \bigcup_{\tau \in S_N} \tau \cdot A.
\]
$S_N \cdot A$ is in $\mathcal S_N$ and is the {\em symmetrized} rectangle formed from $A$. $\mathcal S_N$ is generated by all symmetrized rectangles---in fact, since $\mathbb P_N$ is a Borel measure, we may restrict our attention to symmetrized rectangles $A$ where the $A_m$ are all balls.

\begin{lemma}
\label{push down lemma}
Given $E \subseteq \mathcal S_n$, $\zeta\in \mf o$ and $r \in \N$, define
\[
\zeta + \pi^r E = \{ (\zeta + \pi^r \alpha_1, \ldots, \zeta + \pi^r \alpha_n) : \bs \alpha \in E \}.
\]
Then,
\[
\mathbb P_n(\zeta + \pi^r E) = q^{-r\beta {n \choose 2} - rn} \mathbb P_n(E)
\]
\end{lemma}
\begin{proof}
  Each point in $\zeta + \pi^r E$ can be written in the form $\zeta + \pi^r \bs \alpha$ where $\bs \alpha \in E$. Moreover this map is a bijection. We may thus write
  \[
  \mathbb P_n(\zeta + \pi^r E) = \int_{E} |\Delta_n(\zeta + \pi^r \bs \alpha)|^{\beta} \, d\mu^n(\zeta + \pi \bs \alpha),
  \]
  where
  \[
  |\Delta_n(\zeta + \pi^r \bs \alpha)| = \prod_{\ell < m} |\zeta + \pi^r \alpha_m - (\zeta + \pi^r \alpha_{\ell})| = |\pi|^{r{n \choose 2}} \prod_{\ell < m}|\alpha_m - \alpha_{\ell}| = q^{-r{n \choose 2}} |\Delta_n(\bs \alpha)|.
  \]
  The translation invariance of $\mu$ implies that
  \[
  d\mu^n(\zeta + \pi^r \bs \alpha) = |\pi|^{rn} d \mu^n(\bs \alpha) = q^{-rn} d \mu^n(\bs \alpha).
  \]
  Putting this all together, we see that $\mathbb P_n(\zeta + \pi^r E) = q^{-r\beta {n \choose 2} - rn} \mathbb P_n(E)$.
\end{proof}

\subsection{The $\sigma$-algebra of cylinder sets}
Given $B \in \mathcal B(\mf o)$ we define $N_B: \mf o^N \rightarrow \Z_{\geq 0}$ by $N_B(\bs \alpha) = \# \{\alpha_1, \ldots, \alpha_N\} \cap B$. That is $N_B(\bs \alpha)$ is the number of coordinates of $\bs \alpha$ in $B$. Put another way, $N_B$ is a random variable counting the number of particles in $B$.
We call $N_B^{-1}(n) \subseteq \mf o^N$ a {\em simple} cylinder set. Clearly $N_B^{-1}(n)$ is in $\mathcal S_N$. The $\sigma$-algebra generated by all simple cylinder sets is called the {\em cylinder} $\sigma$-algebra and denoted $\mathcal C_N$. More specifically,
\[
\mathcal C_N = \sigma\{N_B : B \in \mathcal B(\mf o)\} \subseteq \mathcal S_N.
\]
We also define $\mathcal C_{N}(B) = \sigma(N_B) \subseteq \mathcal C_N$ which contains information (only) about the number of particles in $B$. Note that $\mathcal C_{N}(B) = \mathcal C_{N}(B^c)$ is a finite $\sigma$-algebra.

To get a feeling for what cylinder sets look like, fix $B \subseteq \mf o$ and note that
\[
\{N_B = n\} = S_N \cdot \underbrace{B \times \cdots \times B}_n \times \underbrace{B^c \times \cdots \times B^c}_{N-n}
\]
Likewise, if $B$ and $E$ are disjoint sets and $n$ and $m$ non-negative integers with $n + m \leq N$, then
\[
\{N_B =n, N_E = m\} = S_N \cdot \underbrace{B \times \cdots \times B}_n \times \underbrace{E \times \cdots \times E}_m \times \underbrace{(B\cup E)^c \times \cdots \times (B\cup E)^c}_{N-n-m},
\]
and this pattern continues. Suppose $\mathbf B = (B_1, \ldots, B_M)$ is a finite open cover of $\mf o$, and $\mathbf n = (n_1, \ldots, n_M)$ is an {\em occupation vector} such that $n_1 + \cdots + n_M = N$. We define the event
\[
\{N_{\mathbf B} = \mathbf n\} = \{N_{B_1} = n_1, \ldots, N_{B_M} = n_M\} = S_N \cdot B_1^{n_1} \times \cdots \times B_M^{n_M}.
\]
All cylinder sets can be described (via union and intersection) in terms of such sets.

We will be especially interested in such cylinder sets where each of the $B_m$ is a ball. Each $B_m$ can then be written as $\zeta_m + \pi^{r_m} \mf o$ where $\zeta_m \in B_m$ and $r_m$ is a non-negative integer. The strong triangle inequality implies that any $\zeta$ in $B_m$ has equal claim to being its center, and we use this fact to our advantage in the proof of the next theorem.
\begin{thm}
  \label{complete cyl thm}
  Let $(B_m = \zeta_m + \pi^{r_m} \mf o$ : $m = 1, \ldots, M)$ be a collection of disjoint balls, and let $\mathbf n$ be an occupation vector which sums to $N$. Then,
  \begin{align*}
  \mathbb P_N\{N_{\mathbf B} = \mathbf n\} &= \frac{N!}{Z(N, \mf o, \beta)}\bigg\{\prod_{k < \ell}^M |\zeta_{\ell} - \zeta_k|^{\beta n_{\ell} n_k} \bigg\} \prod_{m=1}^M \frac{Z(n_m, \pi^{r_m} \mf o, \beta)}{n_m!} \\
  &= \frac{N!}{Z(N, \mf o, \beta)}\bigg\{\prod_{k < \ell}^M |\zeta_{\ell} - \zeta_k|^{\beta n_{\ell} n_k} \bigg\} \prod_{m=1}^M q^{-r_m \beta {n_m \choose 2} - r_m n_m}\frac{Z(n_m, \mf o, \beta)}{n_m!}.
  \end{align*}
\end{thm}
Before the proof, a couple of remarks are in order. First, this gives an exact method of computing the probabilities of these special cylinder sets in finite time. In fact, since any finitely-described cylinder set is a disjoint union of these special cylinder sets, in fact we can now compute exactly, and in finite time, the probability for any cylinder set we may care about. Finally we remark that the $\prod_{k < \ell}^M |\zeta_{\ell} - \zeta_k|^{\beta n_m n_k}$ term that appears is the Boltzmann factor for a system at inverse temperature $\beta$ and a particle at each of the $\zeta_m$ with integer charge $n_m$. This connects to the {\em multi-component ensemble} where we allow particles to have different integer multiple charges, and is considered in Section~\ref{multi-comp}.
\begin{proof}
There are $N \choose n_1, \ldots, n_M$ images of $B_1^{n_1} \times \cdots \times B_M^{n_M}$ under the action of $S_N$. Thus,
\begin{align*}
  \mathbb P_N\{N_{\mathbf B} = \mathbf n\} &= \frac{1}{Z(N, \mf o, \beta)} \frac{N!}{n_1! \cdots n_M!} \int_{B_1^{n_1}} \cdots \int_{B_M^{n_M}}
\bigg\{\prod_{k < \ell}^M
\prod_{i=1}^{n_{\ell}} \prod_{j=1}^{n_k} |\alpha^{\ell}_i - \alpha^k_j|^{\beta} \bigg\} \\
& \hspace{4cm}\times \prod_{m=1}^M |\Delta_{n_m}(\bs \alpha^m)|^{\beta} d\mu^{n_1}(\bs \alpha^1) \cdots d\mu^{n_M}(\bs \alpha^M).
\end{align*}
But now, if $\alpha_{\ell} \in B_{\ell}$ and $\alpha_k \in B_k$ then,
\[
|\alpha_{\ell} - \alpha_k| = |\alpha_{\ell} - \zeta_{\ell} + \zeta_{\ell} - \zeta_k + \zeta_k - \alpha_k| \leq \max\{|\alpha_{\ell} - \zeta_{\ell}|,|\zeta_{\ell} - \zeta_k|,|\zeta_k - \alpha_k|\}.
\]
Notice also that since $\zeta_k \not \in B_{\ell}$ we must have
\[
|\zeta_{\ell} - \zeta_k| > |\alpha_{\ell} - \zeta_{\ell}| \qq{and similarly}
|\zeta_{\ell} - \zeta_k| > |\alpha_{k} - \zeta_{k}|.
\]
The strong triangle inequality is thus an {\em equality}, and
$|\alpha_{\ell} - \alpha_k| = |\zeta_{\ell} - \zeta_k|$. It follows that
\[
\prod_{k < \ell}^M \prod_{i=1}^{n_{\ell}} \prod_{j=1}^{n_k} |\alpha^{\ell}_i - \alpha^k_j|^{\beta} = \prod_{k < \ell}^M |\zeta_{\ell} - \zeta_k|^{\beta n_{\ell} n_k},
\]
and
\[
\mathbb P_N\{N_{\mathbf B} = \mathbf n\} = \frac{N!}{Z(N, \mf o, \beta)}\bigg\{\prod_{k < \ell}^M |\zeta_{\ell} - \zeta_k|^{\beta n_m n_k} \bigg\} \prod_{m=1}^M \frac{Z(n_m, B_m, \beta)}{n_m!}.
\]
To arrive at the expressions given in the statement of the theorem, we first notice that translating the ball $B_m = \zeta_m + \pi^{r_m} \mf o$ to $\pi^{r_m} \mf o$ does not change the partition function. That is $Z(n_m, B_m, \beta) = Z(n_m, \pi^{r_m} \mf o, \beta)$. The second expression follows from the first by using the fact that $\pi^{r_m} \mf o$ is a contraction of $\mf o$ and the integral defining $Z(n_m, \pi^{r_m} \mf o, \beta)$ can be expressed in terms of $Z(n_m, \mf o, \beta)$ accordingly.
\end{proof}

To derive probabilities of more general cylinder sets $\{N_{\mathbf B} = \mathbf n\}$ where $n_1 + \cdots + n_M = N' < N$, we note that $B$ being a finite union of balls implies that $B^c$ too can be expressed as the disjoint union of finitely many balls in $\mf o$. That is there exist balls $C_1, \ldots, C_K$ so that $B^c = C_1 \sqcup \cdots \sqcup C_K$. There is more than one way of doing this, but there is a unique set of balls (up to reordering) with minimal $K$. We then write
\[
\{N_{\mathbf B} = \mathbf n\} = \bigsqcup_{\Sigma \ell_k = N-N'} \{N_{\mathbf B} = \mathbf n, N_{\mathbf C} = \bs \ell\},
\]
where the union is over all occupation vectors $\bs \ell = (\ell_1, \ldots, \ell_K)$ summing to $N - N'$. Events like $\{N_{\mathbf B} = \mathbf n, N_{\mathbf C} = \bs \ell\}$ are now computable by Theorem~\ref{complete cyl thm}.

\begin{figure}[h!]
\centering
\includegraphics[scale=.4]{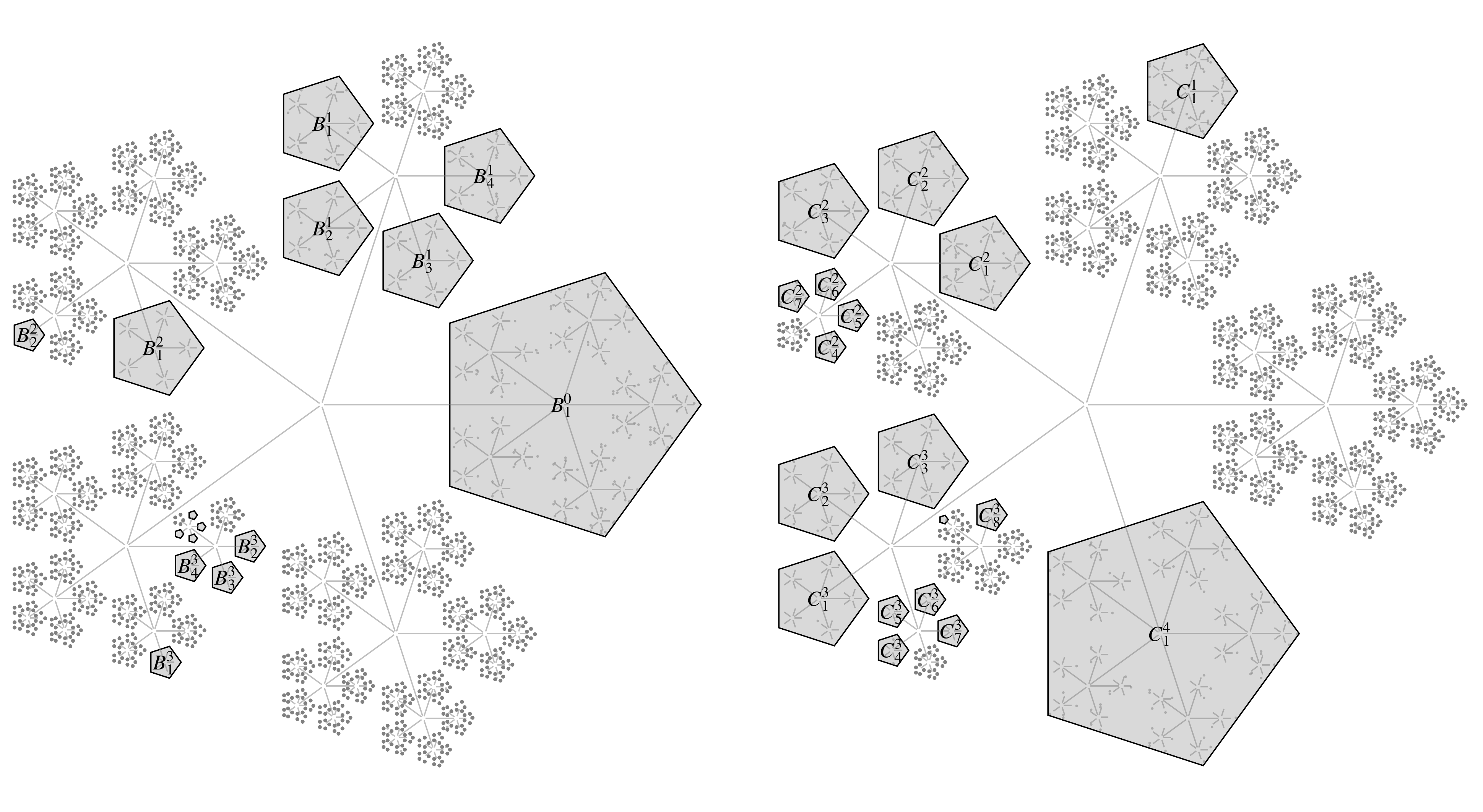}
\caption{An example of a finite disjoint union of balls $\mathbf B$ in $\mf o = \mathbb Z_5$, and its complementary disjoint union $\mathbf C$. These are labeled as in Corollary~\ref{cyl cor 2}, with labels suppressed for smaller diameter balls.}
\end{figure}

\begin{cor}
\label{cyl cor}
Suppose $B = B_1 \sqcup \cdots \sqcup B_M$ is a proper subset of $\mf o$ with each $B_m = \zeta_m + \pi^{r_m} \mf o$, and $B^c = C_1 \sqcup \cdots \sqcup C_K$ where each $C_k = \xi_k + \pi^{t_k} \mf o$. Then, if $\mathbf n = (n_1, \ldots, n_M)$ is an occupation vector with $n_1 + \cdots + n_M = N' < N$. Then,
\begin{align*}
\mathbb P_N\{N_{\mathbf B} = \mathbf n\} &= \frac{N!}{Z(N, \mf o, \beta)}  \bigg[ \frac{Z(N', \mf o, \beta)}{N'!}\mathbb P_{N'}\{N_{\mathbf B} = \mathbf n\} \\ & \sum_{\Sigma \ell = N-N'} \bigg\{ \prod_{m=1}^M \prod_{k=1}^K |\zeta_m - \xi_k|^{\beta n_m \ell_k}\bigg\} \frac{Z(N-N', \mf o, \beta)}{(N-N')!} \mathbb P_{N-N'}\{N_{\mathbf C} = \bs \ell\} \bigg], \nonumber
\end{align*}
where the sum is over all occupation vectors $\bs \ell = (\ell_1, \ldots, \ell_K)$ summing to $N - N'$.
\end{cor}

In many situations we may simplify Corollary~\ref{cyl cor} by using the additivity of energy over cosets: If $\zeta + \pi^r \mf o$ and $\xi + \pi^s \mf o$ are in different cosets of $\mf m$, then $|\zeta - \xi| = 1$. This is the basis for the following formidable-looking simplification.
\begin{cor}
  \label{cyl cor 2}
  Suppose for each $j=0,\ldots,q-1$, $\mathbf B^j = (B^j_1, \ldots, B^j_{M_j})$ is a set of disjoint balls in $j + \mf m$ and $\mathbf C^j = (C^j_1, \ldots, C^j_{K_j})$ is a family of disjoint balls whose union $C^j$ is the complement of $B^j = B^j_1 \sqcup \cdots \sqcup B^j_{M_j}$ in $j + \mf m$. Suppose $B^j_m = \zeta^j_m + \pi^{r^j_m} \mf o$ and $C^j_{\ell} = \xi^j_{\ell} + \pi^{t^j_{\ell}} \mf o$. Set $N^j = n^j_1 + \cdots + n^j_{M_j}$ and suppose $N' = N^0 + \cdots + N^{q-1} \leq N$, and $L = N - N'$.
  \begin{align*}
\mathbb P_N\{N_{\mathbf B^0} = \mathbf n^0, \ldots, N_{\mathbf B^{q-1}} = \mathbf n^{q-1}\} &= \frac{N!}{Z(N, \mf o, \beta)} \bigg\{\prod_{j=0}^{q-1} \frac{Z(N_j, \mf o, \beta)}{N_j!} \mathbb P_{N_j}\{N_{\mathbf B^j} = \mathbf n^j\} \bigg\} \\
& \hspace{-4cm} \sum_{\Sigma L^j = L} \sum_{\Sigma l^0_k = L^0} \cdots \sum_{\Sigma l^{q-1}_k = L^{q-1}} \left\{ \prod_{j=0}^{q-1} \bigg\{\prod_{m=1}^{N_j} \prod_{k=1}^{L_j} | \zeta_m^j - \xi_k^j |^{\beta n_m^j \ell_k^j} \frac{Z(L_j, \mf o, \beta)}{L_j!}\bigg\} \mathbb P_{L_j}\{N_{\mathbf C^j} = \bs \ell^j\} \right\},
\end{align*}
where the sums are over $L^0 + \cdots + L^{q-1} = L$ and all $\bs \ell^j = (\ell^j_1, \ldots, \ell^j_{K_j})$ where $\ell^j_1 + \cdots + \ell^j_{K_j} = L^j$ for $j=0, \ldots, q-1$.
\end{cor}
Note that any cylinder set given as in Corollary~\ref{cyl cor}  can be written as $\{N_{\mathbf B^0} = \mathbf n^0, \ldots, N_{\mathbf B^{q-1}} = \mathbf n^{q-1}\}$ for an appropriate choice of the $\mathbf B^j$ and $\mathbf n^j$.

\subsection{Conditioning on $\mathcal C_N(B)$}

Now that we understand the $\sigma$-algebra of cylinder sets, at least in principle, we turn to the conditional distribution of particles given events like $\{N_B = n\}$. Specifically, we consider events of the form
\[
S_N \cdot A_1 \times \cdots \times A_n \times \underbrace{B^c \times \cdots \times B^c}_{N-n},
\]
where $A_1 \times \cdots \times A_n$ is a measurable rectangle in $B^n$. We define $\mathcal L_{N}(B)$ to be the {\em local} $\sigma$-algebra generated by all such sets (over all possible $0 \leq n \leq N$). We also denote the set $\mathcal L_{N}^n(B)$ by
\[
\mathcal L_{N}^n(B) = \{ E \cap \{N_B = n\} : E \in \mathcal L_{N,B}\}
\]
Events in $\mathcal L_{N}^n(B)$ are events on which $N_B = n$ and which contain information about those $n$ particles in $B$. We remark that $\mathcal L_{N}^n(B)$ is itself {\em not} a $\sigma$-algebra. However, it {\em is} the image of the Borel $\sigma$-algebra on $B^n$ under the map
\[
A \mapsto S_N \cdot A \times (B^c)^{N-n}
\]
In fact, we could take $A \subseteq B^n$ to be symmetrized with respect to $S_n$ and we get a bijective correspondence between $\mathcal L_N^n(B)$ and the $\sigma$-algebra of symmetrized Borel subsets of $B^n$ denoted $\mathcal S_n(B)$. The complementary information, about the particles in $B^c$, is given by the $\sigma$-algebras $\mathcal L_{N,B^c}$ and $\mathcal S_{N-n}(B^c)$ (which is in correspondence with $\mathcal L_{N}^{N-n}(B^c)$).

Our main result for this section is the following.
\begin{thm}
  \label{cond-dist-thm}
  Let $B \subseteq \mf o$ be a ball. Then $\mathcal L_N(B)$ and $\mathcal L_N(B^c)$ are conditionally independent given $\mathcal C_N(B)$. That is, $\mathcal L_N(B) ] \mathcal C_N(B) [ \mathcal L_N(B^c)$. Moreover, the conditional distribution on $\mathcal L_N(B)$ given the event $\{N_B = n\}$ is
  \[
  \mathbb P_N\{ S_N \cdot A_1 \times \cdots A_n \times (B^c)^{N-n} | N_B = n\} = \mathbb P_n(S_n \cdot A_1 \times \cdots \times A_n).
  \]
\end{thm}
\begin{proof}
On the event $\{N_B = n\}$, a generic set $\mathcal L_N^n(B)$ is a union of events which look like $A = S_N \cdot A_1 \times \cdots \times A_n \times (B^c)^{N-n}$. Similarly, a generic event in $\mathcal L_N^{N-n}(B^c)$ looks like $E = S_N \cdot B^n \times E_1 \times \cdots \times E_{N-n}$. The intersection of $A$ and $E$ is then given by $S_N \cdot A_1 \times \cdots \times A_n \times E_1 \times \cdots \times E_{N-n}$.

It follows that
\begin{align*}
\mathbb E[\bs 1_{\{N_B=n\}} \bs 1_A \bs 1_E] &= \frac{N!}{Z(N,\mf o,\beta)}  \int_{A_1} \cdots \int_{A_n} \int_{E_1} \cdots \int_{E_{N-n}} \\ & \hspace{3cm} |\Delta_n(\bs \alpha)|^{\beta} |\Delta_{N-n}(\bs \gamma)|^{\beta} \prod_{j=1}^n \prod_{\ell=1}^{N-n} |\alpha_j - \gamma_{\ell}|^{\beta} \, d\mu^n(\bs \alpha) d\mu^{N-n}(\bs \gamma).
\end{align*}
Suppose now $\alpha \in B$ and $\gamma \in B^c$. Since $B$ is assumed to be a ball, there is some radius on which $B = \{ \alpha' : |\alpha' - \alpha| \leq r\}$. Since $\gamma \not \in B$ we must also have $|\alpha - \gamma| > r$. It follows then that if $\alpha'$ is any other element of $B$,
\[
|\alpha' - \gamma| = |\alpha' - \alpha + \alpha - \gamma| \leq \max\{|\alpha - \alpha'|, |\alpha - \gamma|\} = |\alpha - \gamma|
\]
by the strong triangle inequality. But since $|\alpha - \alpha'| < |\alpha - \gamma|$, in fact the inequality becomes an equality, and we conclude $|\alpha' - \gamma| = |\alpha - \gamma|$. It follows that, as a function of $\bs \alpha$,
\[
\prod_{j=1}^n \prod_{\ell=1}^{N-n} |\alpha_j - \gamma_{\ell}|
\]
is constant on $B^n$, and we write $R_B(\bs \gamma)$ for the resulting function on $(B^c)^{N-n}$. Consequently,
\begin{align*}
\mathbb E[\bs 1_{\{N_B=n\}} \bs 1_A \bs 1_E] &= \frac{N!}{Z(N,\mf o,\beta)}  \int_{A_1} \cdots \int_{A_n} |\Delta_n(\bs \alpha)|^{\beta} d\mu^n(\bs \alpha) \int_{E_1} \cdots \int_{E_{N-n}}  R_B(\bs \gamma)^{\beta} |\Delta_{N-n}(\bs \gamma)|^{\beta} d\mu^{N-n}(\bs \gamma).
\end{align*}
Similar reasoning shows
\begin{align*}
\mathbb E[\bs 1_{\{N_B=n\}} \bs 1_A] &= \frac{N!}{(N-n)! Z(N,\mf o,\beta)}  \int_{A_1} \cdots \int_{A_n} |\Delta_n(\bs \alpha)|^{\beta} d\mu^n(\bs \alpha) \int_{(B^c)^{N-n}}  R_B(\bs \gamma)^{\beta} |\Delta_{N-n}(\bs \gamma)|^{\beta} d\mu^{N-n}(\bs \gamma), \\
\mathbb E[\bs 1_{\{N_B=n\}} \bs 1_E] &= \frac{N!}{n! Z(N,\mf o,\beta)}  \int_{B^n}  |\Delta_n(\bs \alpha)|^{\beta} d\mu^n(\bs \alpha) \int_{E_1} \cdots \int_{E_{N-n}}  R_B(\bs \gamma)^{\beta} |\Delta_{N-n}(\bs \gamma)|^{\beta} d\mu^{N-n}(\bs \gamma),
\end{align*}
and
\[
\mathbb E[\bs 1_{\{N_B=n\}}] = \frac{N!}{n!(N-n)! Z(N,\mf o, \beta)} \int_{B^n}  |\Delta_n(\bs \alpha)|^{\beta} d\mu^n(\bs \alpha) \int_{(B^c)^{N-n}}  R_B(\bs \gamma)^{\beta} |\Delta_{N-n}(\bs \gamma)|^{\beta} d\mu^{N-n}(\bs \gamma).
\]
It follows that
\[
\frac{\mathbb E[\bs 1_{\{N_B=n\}} \bs 1_A]}{\mathbb E[\bs 1_{\{N_B = n\}}]} = \frac{n!}{Z(n, B, \beta)} \int_{A_1} \cdots \int_{A_n} |\Delta_n(\bs \alpha)|^{\beta} \, d\mu^n(\bs \alpha).
\]
and
\[
\frac{\mathbb E[\bs 1_{\{N_B=n\}} \bs 1_A \bs 1_E]}{\mathbb E[\bs 1_{\{N_B = n\}}]} = \frac{\mathbb E[\bs 1_{\{N_B=n\}} \bs 1_E]}{\mathbb E[\bs 1_{\{N_B = n\}}]} \frac{\mathbb E[\bs 1_{\{N_B=n\}} \bs 1_A]}{\mathbb E[\bs 1_{\{N_B = n\}}]}
\]
That is, $\bs 1_A$ and $\bs 1_{E}$ are independent given the event $\{N_B = n\}$. Since $\mathcal C_N(B)$ consists of disjoint unions of the sets $\{N_B = n\}$ we conclude that $\mathcal L_N(B)$ and $\mathcal L_N({B^c})$ are conditionally independent given $\mathcal C_N(B)$.
\end{proof}
That is, the conditional distribution on $L_N(B)$ is given by
\[
d\mathbb P_N(\bs \alpha | N_B=n) = \frac{|\Delta_n(\bs \alpha)|^{\beta}}{Z(n,B,\beta)} d\mu^n(\bs \alpha).
\]

Note that if $C \subseteq B^c$, then $\mathcal L_N(C) \subseteq \mathcal L_N({B^c})$. It follows that if we condition on the event $\{N_B = n, N_C = m\}$ the distribution of the $n$ particles in $B$ is independent of the $m$ particles in $C$. In particular, if $C$ is also a ball, then the conditional density of $\bs \alpha \in B^n$ and $\bs \gamma \in C^m$ is proportional to $|\Delta_n(\bs \alpha)|^{\beta}|\Delta_m(\bs \gamma)|^{\beta}$. We may view this as giving a probability distribution on $B^n \times C^m$ with the $\sigma$-algebra $\mathcal S(B^n) \otimes \mathcal S(C^m)$.

This is the basis for the following corollary.
\begin{cor}
  Suppose $\mathbf B = (B_1, \ldots, B_M)$ is a vector of pairwise disjoint balls in $\mf o$, and let $\mathbf n = (n_1, \ldots, n_M)$ be a vector of non-negative integers with $n_1 + \cdots + n_M \leq N$. Write
  \[
\{N_{\mathbf B} = \mathbf n\} = \{N_{B_1} = n_1, \ldots, N_{B_M} = n_M\}.
  \]
  Then the conditional distribution given $\{N_{\mathbf B} = \mathbf n\}$ on  $\mathcal S(B_1^{n_1}) \otimes \cdots \otimes \mathcal S(B_M^{n_M})$ is given by
  \[
  d\mathbb P_{N}(\bs \alpha_1, \ldots, \bs \alpha_M | N_{\mathbf B} = \mathbf n) = \frac{|\Delta_{n_1}(\bs \alpha_1)|^{\beta}}{Z(n_1, B, \beta)} d\mu^{n_1}(\bs \alpha_1) \cdots \frac{|\Delta_{n_M}(\bs \alpha_M)|^{\beta}}{Z(n_M, B, \beta)} d\mu^{n_M}(\bs \alpha_M).
  \]
\end{cor}

\section{The Grand Canonical Ensemble}

We now allow the system to exchange not only energy with the reservoir but also particles.  That is $N$ is no longer fixed.  In this {\em grand canonical ensemble}, there is a parameter $\chi$ called the {\em chemical potential} which represents the energy cost/reward per particle. In this setting, the energy of a system with $N$ particles in $\mf o$ is given by
\[
E(N, \bs \alpha) = \chi N - \sum_{n < m} \log|\alpha_n - \alpha_m|.
\]
The density of states is then given by
\[
\frac{1}{Z} e^{-\beta E(N, \bs \alpha)}
\]
where, the grand canonical partition function is given byte
\[
Z = Z(\chi, V, \beta) = \sum_{N=0}^{\infty} \frac{1}{N!} Z(N, V, \beta) e^{\chi \beta N}.
\]
Notice the introduction of the $N!$ term in the summand.  This term is now necessary since we are assuming that our particles are indistinguishable.\footnote{Even though the particles are indistinguishable in the canonical ensemble, there the $N!$ term is unnecessary since it appears in both the Boltzmann factor and the partition function and hence cancels in the probability density of states.}
Sometimes, instead of the chemical potential, the energy cost per particle is encoded in the {\em fugacity} parameter $t := e^{-\beta \chi}$ and we write
\[
Z(t, V, \beta) = \sum_{N=0}^{\infty} \frac{1}{N!} Z(N, V, \beta) t^N.
\]
for this version of the grand canonical potential function. We will view $t$ as independent of $\beta$, so that we may also view $Z(t, V, \beta)$ as the exponential generating function for the $Z(N, V, \beta)$. The physical partition function is then given by $Z(e^{-\beta \chi}, V, \beta)$, and when deriving physical quantities, the latter should be used.

Note that, if $\bs \alpha \in \mf o^N$, then the strong triangle inequality implies that $|\Delta(\bs \alpha)| \leq 1$.  Thus, for fixed $\beta \geq 0$,
\begin{align*}
Z(t, V, \beta) &= \sum_{N=0}^{\infty} \frac{t^N}{N!} \int_{V^N} |\Delta(\bs \alpha)|^{\beta} \, d\mu^N(\bs \alpha) \\
&\leq \sum_{N=0}^{\infty} \frac{t^N}{N!} = e^t,
\end{align*}
with equality when $\beta = 0$.  That is, $Z(t, V, \beta)$ is an entire function of $t$ for all $\beta \geq 0$.

\subsection{The Partition Function in the Grand Canonical Ensemble}

First we consider the zero temperature ($\beta = +\infty$) and infinite temperature ($\beta = 0$) regimes. In these cases we can compute the partition function exactly.
\begin{thm}
  \label{edge case}
At infinite temperature,
\[
Z(t, \mf o, 0) = e^t \qq{and} Z(t, \pi^{\ell} \mf o, 0) = e^{t/q^{\ell}}
\]
At zero temperature,
\[
Z(t, \mf o, +\infty) = \left(1 + \frac{t}{q}\right)^q \qq{and} Z(t, \pi^{\ell} \mf o, +\infty) = 1 + \frac{t}{q^{\ell}}
\]
\end{thm}
\begin{proof}
The $\beta = 0$ situation is obvious, since $Z(N, \mf o, 0) = 1$ and $Z(N, \pi^{\ell} \mf o, 0) = q^{-\ell N}$ for all $N$. When $\beta = +\infty$ only states with zero energy contribute to the partition function. Zero energy states can only occur when each of the particles is in a different coset of $\mf m$. Since there are $q$ cosets we must have $0 \leq N \leq q$. For such an $N$ there are ${q \choose N}$ ways of choosing which cosets get the particles, and $N!$ ways of distributing the particles amoungst these cosets. This collection of occupied cosets has volume $1/q^N$ and hence,
\[
Z(t, \mf o, +\infty) = \sum_{N=0}^{q} {q \choose N} \left(\frac{t}{q}\right)^N = \left(1 + \frac{t}{q}\right)^q.
\]
For the neighborhood $\pi^{\ell} \mf o$, the energy is only zero if there are 0 or 1 particles. Since $Z(0, V, +\infty) = 1$ and $Z(1, V, +\infty) = \mu(V)$ the result follows.
\end{proof}

From Theorems~\ref{recursionthm} and \ref{recursion2thm} we quickly find the following
\begin{thm}
  \label{gcz}
  \[
  Z(t,\mf o, \beta) = Z(t, \mf m, \beta)^q.
  \]
\end{thm}
\begin{proof}
  \begin{align*}
    Z(t, \mf o, \beta) &= \sum_{N=0}^{\infty} Z(N, \mf o, \beta) \frac{t^N}{N!} \\
    &= \sum_{N=0}^{\infty} \sum_{\mathbf n} \bigg\{\prod_{r=0}^{q-1} Z(n_r, \mf m, \beta) \frac{t^{n_r}}{n_r!}
    \bigg\}.
  \end{align*}
The inner sum over $\mathbf n$ depends on $N$ since we require that $n_0 + \cdots + n_{q-1} = N$. However, since we are summing over all $N$ we may replace the double sum over $N$ and $\mathbf n$ with a sum over all $q$-tuples of non-negative integers. That is,
\begin{align*}
  Z(t, \mf o, \beta) &= \sum_{n_0} \cdots \sum_{n_{q-1}} \bigg\{\prod_{r=0}^{q-1} Z(n_r, \mf m, \beta) \frac{t^{n_r}}{n_r!}
  \bigg\}.
\end{align*}
But then, Fubini's theorem (factoring the product over the sum) implies that
\begin{align*}
  Z(t, \mf o, \beta) &= \sum_{n_0} Z(n_0, \mf m, \beta) \frac{t^{n_0}}{n_0!} \cdots \sum_{n_{q-1}} Z(n_{q-1}, \mf m, \beta) \frac{t^{n_{q-1}}}{n_{q-1}!} = Z(t, \mf m, \beta)^q. \qedhere
\end{align*}
\end{proof}
There is a physical explanation for this formula. The quantity
\[
\Phi(t, \mf o, \beta) = -\frac{1}{\beta} \log Z(t, \mf o, \beta)
\]
is called the grand canonical {\em potential} and plays the role of free energy in the canonical ensemble. Since it is an energy, we expect that if we have two non-interacting systems that we formally join into a single system, then their grand canonical potentials add. In our situation, the system with a variable number of particles in $\mf o$ can be expressed as the union of $q$ independent systems each with a variable number of particles in a coset of $\mf m$. But the interaction between particles within a coset is independent of the identity of the coset, so we may as well assume that we have $q$ independent copies of an ensemble with a variable number of particles in $\mf m$. The grand canonical potentials add, and thus we have
\[
\Phi(t, \mf o, \beta) = q \cdot \Phi(t, \mf m, \beta)
\]
which leads directly to Theorem~\ref{gcz}.

\begin{proof}[Proof of Theorem~\ref{quad-rec}]
This proof follows from an observation of Cauchy about the relationship between two power series, one of which is an integer power of the other.
\end{proof}

\subsection{Algebraic Formalism for Grand Canonical Partition Functions}

We have seen that many important probabilities and physical quantities of our system reduce to expressions involving $Z(t, \mf m, \beta)$. In particular, our main result $Z(t, \mf o, \beta) = Z(t, \mf m, \beta)^q$ demonstrates the primacy of the quantity $Z(t, \mf m, \beta)$. Since $\mf m$ and $\mf o$ are homeomorphic and isomorphic, we might be tempted to think that $Z(t, \mf m, \beta)$ is itself a $q$ power of another function ($Z(t, \pi^2 \mf o, \beta)$ springs to mind). This is not the case, at least not in the traditional way.

To see why, consider a system with all particles in $\mf m$. By counting the number of particles in each coset of $\pi^2 \mf o$ we get an occupation vector $\mathbf n$. Two particles in different cosets of $\pi^2 \mf o$ will be distance $1/q$ from each other, and each such pair will contribute $\log q$ to the energy of the system. That is, if $(\bs \alpha_0, \ldots, \bs \alpha_{q-1})$ is a factored state of the system so that $\bs \alpha_r$ is an $n_r$ vector whose coordinates are the location of particles in $\pi(r + \mf m)$, then
\begin{equation}
\label{menergy}
E(\bs \alpha_0, \ldots, \bs \alpha_{q-1}) = \sum_{0 \leq r < s < q} n_r n_s \log q + \sum_{r=0}^{q-1} E(\bs \alpha_r)
\end{equation}

In contrast, if two particles are in different cosets of $\mf m$ then they are distance 1 from each other, and there is no interaction energy, $\log q$, introduced by such pairs.

We introduce some algebraic formalism to deal with the terms like $q^{-\beta n_r n_s}$ which arise for the factors introduced into $Z(N, \mf m, \beta)$ from the interaction energy between pairs in different cosets of $\pi^2 \mf o$.

Let $R$ be a commutative ring and let $(R[[t]], +)$ be the group of formal power series in the indeterminant $t$ with the usual definition of $+$. Given two series $\sum a_n t^n$ and $\sum b_m t^m$ and non-zero constant $C \in R$, we define their $\star$-product to be
\[
\sum a_n t^n \star \sum b_m t^m = \sum_{m,n} C^{n m} a_n b_m t^{n+m}.
\]
\begin{lemma}
  \label{lemma:3}
  $\star$ is an associative operation on $R[[t]]$.
\end{lemma}
\begin{proof}
Consider
\begin{align*}
\bigg(\sum a_n t^n \star \sum b_m t^m\bigg) \star \sum c_{\ell} t^{\ell} &= \bigg( \sum_{m,n} C^{n m} a_n b_m t^{n+m} \bigg) \star \sum c_{\ell} t^{\ell}.
\end{align*}
Simplify
\begin{align*}
\sum_{m,n} C^{n m} a_n b_m t^{n+m} = \sum_k \bigg[ \sum_{m+n=k} C^{n m} a_n b_m \bigg] t^k
\end{align*}
Hence,
\begin{align*}
\bigg(\sum a_n t^n \star \sum b_m t^m\bigg) \star \sum c_{\ell} t^{\ell} &=
\sum_{\ell, k} C^{k \ell} c_{\ell} \bigg[ \sum_{m+n=k} C^{n m} a_n b_m \bigg] t^{k + \ell} \\
&= \sum_{\ell, m, n} C^{(n m + n \ell + m \ell)} a_n b_m c_{\ell} t^{n +m + \ell} \\
&= \sum a_n t^n \star \bigg( \sum b_m t^m \star \sum c_{\ell} t^{\ell} \bigg).
\end{align*}
where the conclusion follows from symmetry in the penultimate equation.
\end{proof}

\begin{thm}
$R[[t]]$ is a commutative ring under the operations $+$ and $\star$ with unity $1 = 1 + 0 t + 0 t^2 + \cdots$. Moreover, if $R$ is an integral domain, then so too is $(R[[t]], + ,\star)$.
\end{thm}
\begin{proof}
Commutativity and that $1$ is the multiplicative identity are obvious from the definition of  $\star$. To verify the distributive property of $\star$ over $+$, consider
\begin{align*}
\sum a_n t^n \star \left(\sum b_m t^m + \sum b_m' t^m\right) &= \sum a_n t^n \star \sum (b_m + b_m') t^m \\
&= \sum_{n,m} C^{nm} a_n (b_m + b_m') t^{n+m} \\
&= \sum_{n,m} C^{nm} a_n b_m t^{n+m}+\sum_{n,m} C^{nm} a_n b_m' t^{n+m} \\
&= \sum a_n t^n \star \sum b_m t^m + \sum a_n t^n \star \sum b'_m t^m.
\end{align*}
Next suppose $R$ is an integral domain and $\sum a_n t^n$ and $\sum b_m t^m$ are nonzero. Then there exist $n_0$ and $m_0$ so that $a_{n_0}$ and $b_{m_0}$ are non-zero. It follows that the coefficient of $t^{n_0+m_0}$ in $\sum a_n t^n \star \sum b_m t^m$ is $a_{n_0} b_{m_0} C^{n_0 m_0} \neq 0$. Hence, in this situation $(R[[t]], +, \star)$ has no zero divisors.
\end{proof}

\begin{lemma}
  \label{lemma:2}
In this ring, the $J$th power of a power series satisfies
\[
\bigg(\sum a_n t^n\bigg)^{\star J} = \sum_{n_1, \ldots, n_J} \bigg\{\prod_{j < k}^J C^{n_j n_k} \bigg\} \prod_{j=1}^J a_{n_j} t^{n_j}.
\]
\end{lemma}
\begin{proof}
We induct on $J$. The base case is trivial. Then,
\begin{align*}
\bigg(\sum a_n t^n\bigg)^{\star J} &= \bigg(\sum a_n t^n\bigg)^{\star (J-1)} \star \sum a_n t^n \\
&= \bigg(\sum_{n_1, \ldots, n_{J-1}} \bigg\{\prod_{j < k}^{J-1} C^{n_j n_k} \bigg\} \prod_{j=1}^{J-1} a_{n_j} t^{n_j} \bigg) \star \sum a_n t^n \\
&= \bigg(\sum_m \sum_{n_1 + \cdots + n_{J-1}=m} t^m \bigg\{\prod_{j < k}^{J-1} C^{n_j n_k} \bigg\} \prod_{j=1}^{J-1} a_{n_j} \bigg) \star \sum a_n t^n \\
&= \sum_m \sum_n C^{nm}\bigg(\sum_{n_1 + \cdots + n_{J-1}=m} \bigg\{\prod_{j < k}^{J-1} C^{n_j n_k} \bigg\} \prod_{j=1}^{J-1} a_{n_j}\bigg) a_n t^{n+m} \\
&= \sum_m \sum_n \bigg(\sum_{n_1 + \cdots + n_{J-1}=m} \bigg\{C^{n(n_1 + \cdots + n_{J-1})}\prod_{j < k}^{J-1} C^{n_j n_k} \bigg\} \prod_{j=1}^{J-1} a_{n_j}\bigg) a_n t^{n+m}.
\end{align*}
Renaming the index variable $n$ to $n_J$,
\begin{align*}
\bigg(\sum a_n t^n\bigg)^{\star J} &= \sum_m \bigg(\sum_{n_J} \sum_{n_1 + \cdots + n_{J-1}=m} \bigg\{ \prod_{j < k}^{J} C^{n_j n_k} \bigg\} \prod_{j=1}^{J} a_{n_j} t^{n_j} \bigg).
\end{align*}
The sum over $m$ is superfluous if we remove the restriction on the $n_j$, and we arrive at the formulation in the statement of the lemma.
\end{proof}

\subsection{$\star$ as a convolution operator}

We define a transform on $(R[[t]], +, \star)$ by
\[
\overline{\sum a_n t^n} := \sum C^{n \choose 2} a_n t^n,
\]
where $C^{0 \choose 2} = C^{1 \choose 2} := 1$.
If $C$ is invertible in $R$, then this transform has an inverse given by
\[
\underline{\sum a_n t^n} := \sum C^{-{n \choose 2}} a_n t^n.
\]
Clearly then, using this notation,
\[
\underline{\overline{\sum a_n t^n}} = \sum a_n t^n.
\]
\begin{lemma}
  \label{convolution-lemma}
\[
\overline{\sum a_n t^n} \star \overline{\sum b_m t^n} = \overline{\sum a_n t^n \sum b_m t^m},
\]
where the multiplication on the right hand side is the usual multiplication of power series.
\end{lemma}
\begin{proof}
\begin{align*}
\overline{\sum a_n t^n} \star \overline{\sum b_m t^n} &= \sum C^{n \choose 2} a_n t^n \star \sum C^{m \choose 2} b_m t^m \\
&= \sum_{n,m} C^{n \choose 2} C^{nm} C^{m \choose 2} a_n b_m t^{n + m} \\
&= \sum_N \left(\sum_{n + m = N} C^{n \choose 2} C^{nm} C^{m \choose 2} a_n b_m \right) t^N.
\end{align*}
Finally, since
\[
{n \choose 2} + {m \choose 2} + nm = {n + m \choose 2},
\]
we have
\begin{align*}
\overline{\sum a_n t^n} \star \overline{\sum b_m t^n} &= \sum_N \left(\sum_{n + m = N} a_n b_m \right) C^{N \choose 2} t^N = \overline{\sum a_n t^n \sum b_m t^m}. \qedhere
\end{align*}
\end{proof}

It will be useful to iterate the transform and its inverse. We will denote the $\ell$th iteration of the transform and its inverse (when it exists) by, respectively
\[
\overline{\sum a_n t^n}^{\ell} \qq{and} \underline{\sum a_n t^n}_{\ell}
\]
Note that these are just the transforms formed by replacing $C$ in the original definition with $C^{\ell}$. As such, there is a product $\star^{\ell}$, formed by replacing $C$ in the definition of $\star$ with $C^{\ell}$ so that, for instance
\[
\overline{\sum a_n t^n}^{\ell} \star^{\ell} \overline{\sum b_m t^n}^{\ell} = \overline{\sum a_n t^n \sum b_m t^m}^{\ell}.
\]
To mirror our notation with transforms and inverses, it makes sense to set $\star_{\ell}$ for $\star^{(-\ell)}$. That is $\star_{\ell}$ is the convolution operator formed by replacing $C$ with $C^{-\ell}$. Of course, this only makes sense if $C$ is invertible in $R$. In this situation,
\[
\underline{\sum a_n t^n}_{\ell} \star_{\ell} \underline{\sum b_m t^n}_{\ell} = \underline{\sum a_n t^n \sum b_m t^m}_{\ell}.
\]
Note when $\ell=0$ the transform (and it's inverse) are the identity transform and $\star_0 = \star^0$ is simply the usual multiplication of power series.

\subsection{Partition functions as $\star q$-powers}
In this section, and throughout whenever we are discussing grand canonical partition function, we specify that $C = q^{-\beta}$. The utility of the algebraic constructions we have introduced becomes apparent in the next two results.

\begin{thm}
  \label{gc-thm}
For any Borel set $V$ and $\ell \geq 0$,
\[
Z(t, \pi^{\ell}V, \beta) = \overline{Z}^{\ell}(t/q^{\ell}, V, \beta).
\]
In particular, $Z(t, \mf m, \beta) = \overline{Z}(t/q, \mf o, \beta)$.
\end{thm}
This puts Theorem~\ref{gcz} into a new light, as we can now express the $q$-power relationship between $Z(t, \mf m, \beta)$ and $Z(t, \mf o, \beta)$ as a functional equation that $Z(t, \mf o, \beta)$ must satisfy.
\begin{cor}$Z(t, \mf o, \beta) = \overline{Z}(t/q, \mf o, \beta)^q$. That is $Z(t, \mf o, \beta)$ satisfies $\alpha(t) - \left(\overline{\alpha}(t/q)\right)^q = 0$ in $\R[[t]]$.
\end{cor}

\begin{thm}
  \label{thm4}
  \[
Z(t, \mf m, \beta) = Z(t, \pi^2 \mf o, \beta)^{\star^1 q}.
  \]
More generally, for any integer $\ell$,
\[
Z(t, \pi^{\ell} \mf o, \beta) = Z(t, \pi^{\ell+1} \mf o, \beta)^{\star^{\ell} q}.
\]
\end{thm}
\begin{cor}
  \[
  Z(t, \mf m, \beta) = \overline Z(t/q, \mf m, \beta)^{\star^1 q}
  \]
\end{cor}
Some remarks:
\begin{enumerate}
  \item This says that the grand canonical partition function for any ideal is the $q$th power, using the appropriate operator $\star^{\ell}$, of the grand canonical partition function for its unique maximal ideal;
  \item Theorem~\ref{recursion2thm} corresponds to $\ell = 0$;
  \item This formula is valid for any integer $\ell$, including negative integers. This explains how to extend to fractional ideals outside of $\mf o$;
  \item We may iterate so that, for instance
  \[
  Z(t, \mf o, \beta) = \left(Z(t, \pi^2 \mf o, \beta)^{\star^1 q}\right)^{\star^0 q}.
  \]
\end{enumerate}
\begin{proof}[Proof of Theorem~\ref{thm4}]
In the most general setting, this follows by replacing $\log q$ in \eqref{menergy} with $\ell \log q$. Then, using familiar maneuvers,
\[
Z(N, \pi^{\ell} \mf o, \beta) = \sum_{n_0 + \cdots + n_{q-1}=N} N! \bigg\{\prod_{r < s} q^{-\beta \ell n_r n_s} \bigg\} \prod_{r=0}^{q-1} \frac{1}{n_r!} Z(n_r, \pi^{\ell+1} \mf o, \beta),
\]
and,
\[
Z(t, \pi^{\ell} \mf o, \beta) = \sum_{n_0, \ldots, n_{q-1}} \bigg\{\prod_{r < s} q^{-\beta \ell n_r n_s} \bigg\} \prod_{r=0}^{q-1} \frac{t^{n_r}}{n_r!} Z(n_r, \pi^{\ell+1} \mf o, \beta).
\]
The theorem now follows from Lemma~\ref{lemma:2}.
\end{proof}

\subsection{Probabilities of Cylinder Sets}

We turn to the explicit construction and analysis of the underlying probability space induced by physical considerations. This will be constructed from the probability spaces from the canonical ensemble.

Our probability space
\[
\Omega := \bigsqcup_{N} \mf o^N
\]
with $\sigma$-algebra
\[
\mathcal S = \sigma\{A_1 \sqcup \cdots \sqcup A_M: M \in \mathbb N, A_n \in \mathcal S_n\}.
\]
A generic element in $\mathcal S$ looks like $A = A_1 \sqcup A_2 \sqcup \cdots$ where $A_n \in \mathcal S_n$. The probability measure induced by the Boltzmann factor is then
\[
\mathbb P(A) := \frac{1}{Z(t, \mf o, \beta)} \sum_{N=0}^{\infty} Z(N, A_N, \beta) \frac{t^N}{N!},
\]
or equivalently,
\[
\mathbb P(A | N_{\mf o} = N) = \frac{Z(N, A_N, \beta)}{Z(N, \mf o, \beta)} \qq{and} \mathbb P\{N_{\mf o} = N\} = \frac{Z(N, \mf o, \beta)}{Z(t, \mf o, \beta)} \frac{t^N}{N!}.
\]
$\mathbb P(A)$ depends implicitly on $\beta$ and $t$. If we need to make this dependence explicit, we will write $\mathbb P(t, A, \beta)$. Theorem~\ref{edge case} implies that when $\beta = 0$, $\mathbb P\{N_{\mf o} = N\} = e^{-t} {t^N}/{N!}$, and hence $N_{\mf o}$ is a Poisson random variable with parameter $t$. When $\beta = +\infty$, $N_{\mf o}$ is a binomial random variable with $q$ trials each with probability of success $t/(q + t)$.

The numerator of $\mathbb P(A)$ is an important generating series, and we define
\[
Z(A) = Z(t, A, \beta) := \sum_{N=0}^{\infty}  Z(N, A_N, \beta) \frac{t^N}{N!}.
\]


Given $A$ as above, $\zeta \in \mf o$ and $r \in \N$, define
\begin{equation}
  \label{push down def}
\zeta + \pi^{\ell} A = \bigsqcup_N (\zeta + \pi^{\ell} A_N).
\end{equation}
Like in the canonical ensemble, there is a simple relationship between $Z(t, \zeta + \pi^{\ell} A, \beta)$ and $\overline{Z}^{\ell}(t/q^{\ell}, A, \beta)$. This is recorded in the following lemma, which follows immediately from Lemma~\ref{push down lemma} and the definition of $\overline{\sum a_n t^n}^{\ell}$.
\begin{lemma}
\[
Z(t, \zeta + \pi^{\ell}\mf o, \beta) =  \overline{Z}^{\ell}(t/q^{\ell}, \mf o, \beta).
\]
\end{lemma}

Given Borel $B \subseteq \mf o$ we define $\mathcal C(B) = \sigma(N_B)$ as before, and we define the $\sigma$-algebra of cylinder sets $\mathcal C \subseteq \mathcal S$ to be that generated by $\{N_B: B \in \mathcal B\}$. Each $C$ in $\mathcal C$ can be written as $C = C_1 \sqcup C_2 \sqcup \cdots$ where $C_N = C \cap \{N_{\mf o} = N\} \in \mathcal C_N$. To get a feel for cylinder sets in the grand canonical setting, consider the event $\{N_B = n\}$,
\[
\{N_B = n\} = \bigsqcup_{N=n}^{\infty} S_N \cdot B^n \times (B^c)^{N-n}
\]
The event that {\em all} particles are in $B$ is given by $\bigsqcup_N B^N$.

As in the canonical ensemble, the cylinder sets whose probabilities are most easy to describe are those of the form $\{N_{\mathbf B} = \mathbf n\}$ where $\mathbf B = (B_1, \ldots, B_M)$ is a family of disjoint balls. Here we do not necessarily assume that the union of the balls is all of $\mf o$, but we do relabel the balls so that $\mathbf B^j = (B^j_1, \ldots, B^j_{M_j})$ are the family of balls contained in $j + \mf m$. By likewise reorganizing the coordinates of $\mathbf n$ we can write
\[
\{N_{\mathbf B} = \mathbf n\} = \{N_{\mathbf B_0} = \mathbf n_0, \ldots, N_{\mathbf B_{q-1}} = \mathbf n_{q-1} \}.
\]

We use the fact that $j + \mf m$ is in bijection with $\mf o$ under the map $\alpha \mapsto j + \pi \alpha$, and write $\underline{B}_k^j$ for the pre-image of $B_k^j$ under this map. If $B_k^j$ has radius $q^{-r_k^j}$ then $\underline B_k^j$ has radius $q^{-r_k^j + 1}$. The $\underline{\mathbf B}^j = (\underline{B}_1^j, \ldots, \underline{B}_{M_j}^j)$ are now disjoint balls in $\mf o$.
If we can write $Z(t,\{N_{\mathbf B} = \mathbf n\},\beta)$ in terms of the $Z(t,\{N_{\underline{\mathbf B}^j} = \mathbf n^j\},\beta)$ and provide a formula for $ Z(t, \{N_{\mf o} = n\}, \beta)$ we will have an inductive formula for $Z(t,\{N_{\mathbf B} = \mathbf n\},\beta)$.

As in the canonical case we take $\mathbf C^j = (C^j_1, \ldots, C^j_{K_j})$ to be a collection of disjoint balls complementary to the $\mathbf B^j$ in $j + \mf m$. Clearly then $\underline{\mathbf C}^j$ is complementary to $\underline{\mathbf B}^j$ in $\mf o$. We remark that it is possible that either $\mathbf B^j$ or $\mathbf C^j$ are empty. If, for instance $\mathbf B^j$ is empty, then necessarily $\mathbf n^j$ is empty and $\mathbf C^j = j + \mf m$. We need to then decipher what the event $\{N_{\mathbf B^j} = \mathbf n^j\}$ means when $\mathbf B^j$ and $\mathbf n^j$ are empty. In this case, we are putting no restriction on the number of particles in $j + \mf m$, and this event is equal to
\[
E := \bigsqcup_N (j + \mf m)^N.
\]
Note that $\underline E$ is the probability that all particles are in $\mf o$ with no restrictions on their number or location---that is with no restriction whatsoever. It follows that $\underline E = \Omega$ and $Z(t, \underline E, \beta) = Z(t, \mf o, \beta)$. In essense, this says that if the event $\{\mathbf N_{\mathbf B} = \mathbf n\}$ makes no specifications on the number of particles in a coset, then that coset contributes a factor of $Z(t, \mf m, \beta)$ to $Z(t,\{N_{\mathbf B} = \mathbf n\},\beta)$. Put another way, this provides one base case necessary for an inductive formula for $Z(t, \{N_{\mathbf B} = \mathbf n\}, \beta)$ in terms of the $Z(t,\{N_{\underline{\mathbf B}^j} = {\mathbf n}^j\},\beta)$. The other base case is given by
\[
Z(t, \{N_{\mf o} = n\}, \beta) = \frac{t^N}{N!} Z(n, \mf o, \beta).
\]

\begin{thm}
  \label{gc-cyl-thm}
  \[
  Z(t,\{N_{\mathbf B} = \mathbf n\},\beta) = \prod_{j=0}^{q-1} \overline{Z}(t/q, \{N_{\underline{\mathbf B}^j} = \mathbf n^j\}, \beta).
  \]
\end{thm}
\begin{proof}
For each $J$, $\mathbf n^j = (n^j_1, \ldots, n^j_{M_j})$, define $N^j := n^j_1 + \cdots + n^j_{M_j}$. On the event $\{ N_{\mathbf B} = \mathbf n\}$, $N_j$ the minimum number of particles in $j + \mf m$.

and using Corollary~\ref{cyl cor 2},
\begin{align*}
Z(t, \{N_{\mathbf B} = \mathbf n\}, \beta) &= \sum_{N=N'}^{\infty} t^N \bigg\{\prod_{j=0}^{q-1} \frac{Z(N^j, \{N_{\mathbf B^j} = \mathbf n^j\}, \beta)}{N^j!} \bigg\} \\
& \sum_{\Sigma L^j = N-N'} \sum_{\Sigma l^0_k = L^0} \cdots \sum_{\Sigma l^{q-1}_k = L^{q-1}} \left\{ \prod_{j=0}^{q-1} \bigg\{\prod_{m=1}^{N^j} \prod_{k=1}^{L^j} | \zeta_m^j - \xi_k^j |^{\beta n_m^j \ell_k^j} \frac{Z(L^j, \{N_{\mathbf C^j} = \bs \ell^j\}, \beta)}{L^j!}\bigg\} \right\},
\end{align*}
where the sums are over $L^0 + \cdots + L^{q-1} = N-N'$ and all $\bs \ell^j = (\ell^j_1, \ldots, \ell^j_{K_j})$ where $\ell^j_1 + \cdots + \ell^j_{K_j} = L^j$ for $j=0, \ldots, q-1$.

The sum over $N$ and frees the constraints on $L^0 + \cdots + L^{q-1} = N-N'$, and allows us to sum over all non-negative $L^0, \ldots, L^{q-1}$. That is,
\begin{align*}
Z(t,\{N_{\mathbf B} = \mathbf n\}, \beta) &=  t^{N'} \bigg\{\prod_{j=0}^{q-1} \frac{Z(N^j, \{N_{\mathbf B^j} = \mathbf n^j\}, \beta)}{N^j!} \bigg\} \\
& \sum_{L^0, \ldots, L^{q-1}} \sum_{\Sigma l^0_k = L^0} \cdots \sum_{\Sigma l^{q-1}_k = L^{q-1}} \left\{ \prod_{j=0}^{q-1} t^{L^j} \frac{Z(L^j, \{N_{\mathbf C^j} = \bs \ell^j\}, \beta)}{L^j!} \bigg\{\prod_{m=1}^{M^j} \prod_{k=1}^{K^j} | \zeta_m^j - \xi_k^j |^{\beta n_m^j \ell_k^j} \bigg\} \right\}.
\end{align*}
This allows us to exchange the (inside) product over $j = 0,\ldots,q-1$ and the sums over the $L^j$. That is,
\begin{align*}
Z(t, \{N_{\mathbf B} = \mathbf n\}, \beta) &=  t^{N'} \bigg\{\prod_{j=0}^{q-1} \frac{Z(N^j, \{N_{\mathbf B^j} = \mathbf n^j\}, \beta)}{N^j!} \\ & \qquad \times \sum_L \sum_{\Sigma l_k = L}  t^{L} \bigg\{\prod_{m=1}^{M^j} \prod_{k=1}^{K^j} | \zeta_m^j - \xi_k^j |^{\beta n_m^j \ell_k^j} \bigg\} \frac{Z(L, \{N_{\mathbf C^j} = \bs \ell\}, \beta)}{L!} \bigg\} \\
&=  \bigg\{\prod_{j=0}^{q-1} \sum_L t^{N^j + L} \frac{Z(N^j, \{N_{\mathbf B^j} = \mathbf n^j\}, \beta)}{N^j!} \\ & \qquad \times   \sum_{\Sigma l_k = L} \bigg\{\prod_{m=1}^{M^j} \prod_{k=1}^{K^j} | \zeta_m^j - \xi_k^j |^{\beta n_m^j \ell_k^j} \bigg\} \frac{Z(L, \{N_{\mathbf C^j} = \bs \ell\}, \beta)}{L!} \bigg\}.
\end{align*}
Note that $\{N_{\mathbf B^j} = \mathbf n^j\}$ accounts for the whereabouts of $N^j$ particles in $j + \mf m$, and hence
\[
Z(N^j,\{N_{\mathbf B^j} = \mathbf n^j\},\beta) = Z({N^j},\{N_{j + \pi \underline{\mathbf B}^j} = \mathbf n^j\},\beta) = q^{-\beta{N^j \choose 2} - N^j} Z(N^j, \{N_{\underline{\mathbf B}^j} = \mathbf n^j\}, \beta).
\]
Similarly,
\[
Z(L,\{N_{\mathbf C^j} = \bs \ell\},\beta) = q^{-\beta{L \choose 2} - L} Z(L,\{N_{\underline{\mathbf C}^j} = \bs \ell\}, \beta).
\]
And, since
\[
{N^j \choose 2} + {L \choose 2} = {N^j + L \choose 2} - LN^j,
\]
we have
\begin{align*}
& t^{N^j + L} Z(L,\{N_{\mathbf C^j} = \bs \ell\},\beta) Z(N^j, \{N_{\mathbf B^j} = \mathbf n^j\}, \beta) \bigg\{\prod_{m=1}^{M^j} \prod_{k=1}^{K^j} | \zeta_m^j - \xi_k^j |^{\beta n_m^j \ell_k^j} \bigg\} \\
& \quad = q^{-\beta {N^j + L \choose 2}} \left(\frac{t}{q}\right)^{N^j + L} Z(N^j,\{N_{\underline{\mathbf B}^j} = \mathbf n^j\},\beta) Z(L,\{N_{\underline{\mathbf C}^j} = \bs \ell\},\beta) \bigg\{\prod_{m=1}^{M^j} \prod_{k=1}^{K^j} q^{\beta} | \zeta_m^j - \xi_k^j |^{\beta n_m^j \ell_k^j} \bigg\}.
\end{align*}
Now note that, if we denote the centers of $\underline{B}^j_{m}$ and $\underline{C}^j_k$ by $\underline{\zeta}^j_m$ and $\underline{\xi}^j_k$, then $|\underline{\zeta}^j_m - \underline{\xi}^j_k| = q|\zeta^j_m - \xi^j_k|$. It follows that
\begin{align*}
& t^{N^j + L} Z(L, \{N_{\mathbf C^j} = \bs \ell\}, \beta) Z(N^j,\{N_{\mathbf B^j} = \mathbf n^j\}, \beta) \bigg\{\prod_{m=1}^{M^j} \prod_{k=1}^{K^j} | \zeta_m^j - \xi_k^j |^{\beta n_m^j \ell_k^j} \bigg\} \\
& \quad = q^{-\beta {N^j + L \choose 2}} \left(\frac{t}{q}\right)^{N^j + L} Z(N^j,\{N_{\underline{\mathbf B}^j} = \mathbf n^j\},\beta) Z(L,\{N_{\underline{\mathbf C}^j} = \bs \ell\},\beta) \bigg\{\prod_{m=1}^{M^j} \prod_{k=1}^{K^j} | \underline{\zeta}_m^j - \underline{\xi}_k^j |^{\beta n_m^j \ell_k^j} \bigg\},
\end{align*}
and hence
\begin{align*}
Z(t, \{N_{\mathbf B} = \mathbf n\}, \beta) &= \bigg\{\prod_{j=0}^{q-1} \sum_L q^{-\beta{N^j + L \choose 2}}\left(\frac{t}{q}\right)^{N^j + L} \frac{Z(N^j, \{N_{\underline{\mathbf B}^j} = \mathbf n^j\}, \beta)}{N^j!} \\ & \qquad \times   \sum_{\Sigma l_k = L} \bigg\{\prod_{m=1}^{M^j} \prod_{k=1}^{K^j} | \underline{\zeta}_m^j - \underline{\xi}_k^j |^{\beta n_m^j \ell_k^j} \bigg\} \frac{Z(L, \{N_{\underline{\mathbf C}^j} = \bs \ell\}, \beta)}{L!}  \bigg\}.
\end{align*}
Since the $\underline{\mathbf B}^j$ and $\underline{\mathbf C}^j$ union to all of $\mf o$, we can use Corollary~\ref{cyl cor 2} again to find
\begin{align*}
Z(t, \{N_{\mathbf B} = \mathbf n\}, \beta) = \bigg\{\prod_{j=0}^{q-1} \sum_L \frac{\left(t/q\right)^{N^j + L}}{(N^j + L)!} q^{-\beta{N^j + L \choose 2}} Z(N^j+L, \{N_{\underline{\mathbf B}^j} = \mathbf n^j\}, \beta) \bigg\}
\end{align*}
In fact, since $Z(K, \{N_{\mathbf B^j} = \mathbf n^j\}, \beta) = 0$ if $K < N^j$,
\begin{align*}
Z(t, \{N_{\mathbf B} = \mathbf n\}, \beta) &= \bigg\{\prod_{j=0}^{q-1} \sum_K \frac{(t/q)^{K}}{K!} q^{-\beta{K \choose 2}}Z(K,\{N_{\underline{\mathbf B}^j} = \mathbf n^j\}, \beta) \bigg\}
\end{align*}
where the last equality follows from Theorem~\ref{gc-thm}. That is,
\[
Z(t,\{N_{\mathbf B} = \mathbf n\},\beta) = \prod_{j=0}^{q-1} \overline{Z}(t/q, \{N_{\underline{\mathbf B}^j} = \mathbf n^j\}, \beta).  \qedhere
\]
\end{proof}

As an example, consider the cylinder set $\{ N_\mathbf B = \mathbf n\}$ in the $5$-adics given diagramatically by
\[
\{N_{\mathbf B} = \mathbf n\} = \foo.
\]
In words, $\{ N_\mathbf B = \mathbf n\}$ is the set of states for which there are 6 particles in a fixed coset of $\mf m = 5 \mathbb Z_5$, and 4 particles in a distinct coset of $\pi^2 \mf o = 25 \mathbb Z_5$. In our case, there are three cosets of $\mf m$ whose occupation number is unspecified. Each of these three will contribute a factor of $\overline{Z}(t/5, \mf o, \beta)$. The coset containing 6 particles will contribute a factor of
\[
\overline{Z}(t/5, \{N_{\mf o} = 6\}, \beta) =  Z(6, \mf o, \beta) 5^{-\beta{6 \choose 2}} \frac{(t/5)^6}{6!}.
\]
The remaining coset contributes a factor of
\begin{align*}
\overline{\overline{Z}(t/{25}, \{\N_{\mf m} = 4\}, \beta) \overline{Z}(t/{25}, \mf o, \beta)^4} &=
\overline{Z}^2(t/{25}, \{\N_{\mf o} = 4\}, \beta) \star \overline{Z}^2(t/{25}, \mf o, \beta)^{\star 4} \\
&= Z(4, \mf o, \beta) 5^{-2 \beta{4 \choose 2}} \frac{(t/25)^4}{4!} \star \overline{Z}^2(t/{25}, \mf o, \beta)^{\star 4}.
\end{align*}
Putting it all together, we have
\[
Z(t, \{N_{\mathbf B} = \mathbf n\}, \beta) = \overline{Z}(t/5, \mf o, \beta)^3 Z(6, \mf o, \beta) 5^{-\beta{6 \choose 2}} \frac{(t/5)^6}{6!} \left(Z(4, \mf o, \beta) 5^{-2 \beta{4 \choose 2}} \frac{(t/25)^4}{4!} \star \overline{Z}^2(t/{25}, \mf o, \beta)^{\star 4} \right).
\]

\section{Multi-Component Ensembles}
\label{multi-comp}
Here we allow particles to have different integer charges. Here we suppose $\mathbf Q = (Q_1, \ldots, Q_M)$ be distinct positive integers representing allowable charges. To eliminate notational complexity, we will view $\mathbf Q$ as fixed once and for all.  Let $\mathbf N = (N_1, \ldots, N_M)$ be positive integers representing the number of each species of particles.

If we suppose that the particles of charge $Q_m$ are identified with the coordinates of $\bs \alpha_{m} = (\alpha_m^1, \ldots, \alpha_m^{N_m})\in \mf o^{N_m}$ then the state of the system is specified by $(\bs \alpha_1, \ldots, \bs \alpha_M)$ and the energy of such a system is given by
\begin{equation}
  \label{multi-charge}
E(\bs \alpha_1, \ldots, \bs \alpha_M) = -\sum_{m=1}^M Q_m^2 \sum_{j < k}^{N_m} \log|\alpha_m^j - \alpha_m^k| - \sum_{\ell < m}^M Q_{\ell} Q_m \sum_{j=1}^{N_{\ell}} \sum_{k=1}^{N_m} \log|\alpha_{\ell}^j - \alpha_m^k|.
\end{equation}
This is the sum of the interaction energies between all pairs of particles.

The partition function in the canonical ensemble is thus
\begin{align*}
Z(\mathbf N, \mf o, \beta) &= \int\limits_{\bs \alpha_1 \in \mf o^{N_1}} \cdots
\int\limits_{\bs \alpha_M \in \mf o^{N_M}} \bigg\{ \prod_{m=1}^M \prod_{j<k}^{N_m} |\alpha_m^j - \alpha_m^k|^{\beta Q_m^2} \bigg\} \\
& \qquad \times \bigg\{\prod_{\ell < m}^M \prod_{j=1}^{N_{\ell}} \prod_{k=1}^{N_m} |\alpha_{\ell}^j - \alpha_m^k|^{\beta Q_m Q_{\ell}} \bigg\} d\mu^{N_1}(\bs \alpha^1) \cdots d\mu^{N_M}(\bs \alpha^M),
\end{align*}
and the partition function for the grand canonical ensemble (in the fugacity variables $\mathbf t = (t_1, \ldots, t_M)$) is given by
\[
Z(\mathbf t, \mf o, \beta) = \sum_{\mathbf N} Z(\mathbf N, \mf o, \beta) \frac{\mathbf t^{\mathbf N}}{\mathbf N!} \qq{where} \frac{\mathbf t^{\mathbf N}}{\mathbf N!} = \prod_{m=1}^M \frac{t_m^{N_m}}{N_m!}
\]

\begin{thm}
  \[
Z(\mathbf t, \mf o, \beta) = Z(\mathbf t, \mf m, \beta)^q
  \]
\end{thm}
\begin{proof}
A bit of strategy is in order. As before we will partition particles according to which coset they reside in. The allowance of different charges doesn't change the fact that the interaction energy between particles in different cosets is always zero.  That is, the energy is still additive over cosets. This implies the integrand in $Z(\mathbf N, \mf o, \beta)$ factors over cosets and Fubini's Theorem allows us to move the product, which is indexed by the cosets of $\mf m$, outside the integral. The translation invariance of the remaining integrands allows us to translate each coset to $\mf m$. Some combinatorial reorganization of the sums appearing in the partition function will then produce the result.

Given a state $(\bs \alpha_1, \ldots, \bs \alpha_M)$, for each $1 \leq m \leq M$
we introduce a factored state for $\bs \alpha_m$ given by $(\bs \alpha_m^0, \ldots \bs \alpha_m^{q-1})$. We will denote the $j$th entry of $\bs \alpha_m^r$ be $\alpha_m^r(j)$. We then define $\mathbf N_m := [N_m^0, \ldots, N_m^{q-1}]$ for the vector of non-negative integers counting how many of the particles with charge $Q_m$ are in each of the cosets. We view $\mathbf N$ as an $M \times q$ matrix of non-negative integers the $m, r$ entry of which $N_m^r$ is the number of particles with charge $Q_m$ in coset $r + \mf m$. Note that $N_m^0 + \cdots + N_m^{q-1} = N_m$. In summary $N_m$ is the number of particles with charge $Q_m$. These particles are located at the coordinates of $\bs \alpha_m$. The vector $\mathbf N_m$ gives the number of particles in each coset and $(\bs \alpha_m^0, \ldots, \bs \alpha_m^{q-1})$ is the factored state with $\bs \alpha_m^r \in \mf o^{N_m^r}$ representing the location of all charge $Q_m$ particles lying in the coset $r + \mf m$.

Then,
\begin{equation}
 \label{multi-charge2}
E(\bs \alpha_1, \ldots, \bs \alpha_M) = -\sum_{r=0}^{q-1} \bigg[ \sum_{m=1}^M Q_m^2 \sum_{j < k}^{N_m^r} \log|\alpha_m^{r}(j) - \alpha_m^{r}(k)| - \sum_{\ell < m}^M Q_{\ell} Q_m \sum_{j=1}^{N_{\ell}^r} \sum_{k=1}^{N_m^r} \log|\alpha_{\ell}^{r}(j) - \alpha_m^{r}(k)| \bigg].
\end{equation}

We will reorganize the integral defining the grand canonical partition function replacing the integral over all state vectors with an integral over factored state vectors. In order to do this correctly we need to account for the number of state vectors corresponding to a factored state. By permuting the corrdinates of $(\bs \alpha_m^0, \ldots, \bs \alpha_m^{q-1})$ we arrive at generic state vector specifying the positions of the particles with charge $Q_m$. Thus a choice of $\bs \alpha_m^r$ for all $m =1, \ldots, M$ and $r = 0, \ldots, q-1$ leads to $N_1! \cdots N_M!$ different states of the system. The set of factored states still overcounts unique states, since permuting the coordinates of any one $\bs \alpha_m^r$ does not actually change the state of the system. Thus, in order to compensate for the overcounting within factored states, we need to introduce terms like $N_m^0! \cdots N_m^{q-1}!$ into the denominator. That is, if we integrate over factored states instead of all state vectors we need to compensate with the combinatorial term
\[
\frac{N_1!}{N_1^0! \cdots N_1^{q-1}!} \cdot \cdots \cdot \frac{N_M!}{N_M^0! \cdots N_M^{q-1}!}
\]
Now,
\begin{align*}
Z(\mathbf t, \mf o, \beta) &= \sum_{N_1=0}^{\infty} \frac{t_1^{N_1}}{N_1!} \cdots \sum_{N_M=0}^{\infty} \frac{t_M^{N_M}}{N_M!} \int\limits_{\bs \alpha_1 \in \mf o^{N_1}} \cdots
\int\limits_{\bs \alpha_M \in \mf o^{N_M}} \bigg\{ \prod_{m=1}^M \prod_{j<k}^{N_m} |\alpha^m_j - \alpha^m_k|^{\beta Q_m^2} \bigg\} \\
& \qquad \times \bigg\{\prod_{\ell < m}^M \prod_{j=1}^{N_{\ell}} \prod_{k=1}^{N_m} |\alpha^{\ell}_j - \alpha^m_k|^{\beta Q_m Q_{\ell}} \bigg\} d\mu^{N_1}(\bs \alpha_1) \cdots d\mu^{N_M}(\bs \alpha_M).
\end{align*}
First we note that by integrating over factored states we replace integrals such as
\[
\int\limits_{\bs \alpha_m \in \mf o^{N_m}} \qq{with} \bigg\{ \int\limits_{\bs \alpha_m^0 \in \mf m^{N_m^0}} \cdots \int\limits_{\bs \alpha_m^{q-1} \in \mf (q-1 + \mf m)^{N_m^{q-1}}} \bigg\}
\]
The combinatorial factor together with the fact that $N_m^{0} + \cdots + N_m^{q-1} = N_m$ allows us to replace the sums
\[
\sum_{N_m} \frac{t_m^{N_m}}{N_m!} = \sum_{N_m} \sum_{\{\mathbf N_m : N_m^{0} + \cdots + N_m^{q-1} = N_m\} } \frac{t_m^{N_m^0} \cdots t_m^{N_m^{q-1}}}{N_m^0! \cdots N_m^{q-1}!},
\]
The sum over $N_m$ however, makes the condition that $N_m^{0} + \cdots + N_m^{q-1} = N_m$ superfluous, and since the summand depends in no other way on $N_m$, we can in fact write
\[
\sum_{N_m} \frac{t_m^{N_m}}{N_m!} = \sum_{\mathbf N_m} \frac{t_m^{N_m^0} \cdots t_m^{N_m^{q-1}}}{N_m^0! \cdots N_m^{q-1}!},
\]
where the sum over the $\mathbf N_m$ is no longer constrained.

It follows then that
\begin{align*}
Z(\mathbf t, \mf o, \beta)
& =  \sum_{\mathbf N_1} \frac{t_1^{N_1^{0}} \cdots t_1^{N_1^{q-1}}}{N_1^0! \cdots N_1^{q-1}!} \cdots \sum_{\mathbf N_M} \frac{t_M^{N_M^{0}} \cdots t_M^{N_M^{q-1}}}{N_M^0! \cdots N_M^{q-1}!}\\
& \qquad \times
\bigg\{ \int\limits_{\bs \alpha_1^0 \in \mf m^{N_1^0}} \cdots \int\limits_{\bs \alpha_1^{q-1} \in \mf (q-1 + \mf m)^{N_1^{q-1}}} \bigg\} \cdots \bigg\{ \int\limits_{\bs \alpha_M^0 \in \mf m^{N_M^0}} \cdots \int\limits_{\bs \alpha_M^{q-1} \in \mf (q-1 + \mf m)^{N_M^{q-1}}} \bigg\}
\\
& \qquad \qquad \times \prod_{r=0}^{q-1} \bigg\{ \prod_{m=1}^M \prod_{j < k}^{N_m^r} |\alpha_m^r(j) - \alpha_m^r(k)|^{\beta Q_m^2} \times \prod_{\ell < m}^M \prod_{j=1}^{N_{\ell}^r} \prod_{k=1}^{N_m^r} |\alpha_{\ell}^r(j) - \alpha_m^r(k)|^{\beta Q_{\ell} Q_m} \bigg\}  \\
& \qquad \qquad \qquad \times d\mu^{N^0_1}(\bs \alpha_1^0) \cdots d\mu^{N^{q-1}_1}(\bs \alpha_1^{q-1}) \cdots d\mu^{N^0_M}(\bs \alpha_M^0) \cdots d\mu^{N^{q-1}_M}(\bs \alpha_M^{q-1}).
\end{align*}
The factorization of the integrand over cosets is the multiplicative version of
the additivity of energy over cosets, \ref{multi-charge}.

Next, we reorder the integrals so that, instead of being grouped by the charge $Q_m$, they are instead grouped by cosets. That is,
\begin{align*}
& \bigg\{ \int\limits_{\bs \alpha_1^0 \in \mf m^{N_1^0}} \cdots \int\limits_{\bs \alpha_1^{q-1} \in \mf (q-1 + \mf m)^{N_1^{q-1}}} \bigg\} \cdots \bigg\{ \int\limits_{\bs \alpha_M^0 \in \mf m^{N_M^0}} \cdots \int\limits_{\bs \alpha_M^{q-1} \in \mf (q-1 + \mf m)^{N_M^{q-1}}} \bigg\} \\
& \qquad = \bigg\{ \int\limits_{\bs \alpha_1^0 \in \mf m^{N_1^0}} \cdots \int\limits_{\bs \alpha_M^0 \in \mf m^{N_M^0}} \bigg\} \cdots \bigg\{ \int\limits_{\bs \alpha_1^{q-1} \in (q-1 +\mf m)^{N_1^{q-1}}} \cdots \int\limits_{\bs \alpha_M^{q-1} \in (q-1+\mf m)^{N_M^{q-1}}} \bigg\}
\end{align*}
We are now in position to use Fubini's Theorem,
\begin{align*}
&  Z(\mathbf t, \mf o, \beta)
   =  \sum_{\mathbf N_1} \frac{t_1^{N_1^{0}} \cdots t_1^{N_1^{q-1}}}{N_1^0! \cdots N_1^{q-1}!} \cdots \sum_{\mathbf N_M} \frac{t_M^{N_M^{0}} \cdots t_M^{N_M^{q-1}}}{N_M^0! \cdots N_M^{q-1}!}\\
& \qquad  \times  \prod_{r=0}^{q-1} \int\limits_{\bs \alpha_1 \in (r +\mf m)^{N_1^r}} \cdots \int\limits_{\bs \alpha_M \in (r+\mf m)^{N_M^r}} \prod_{m=1}^M \prod_{j < k}^{N_m} |\alpha_{m,j} - \alpha_{m,k}|^{\beta Q_m^2} \times \prod_{\ell < m}^M \prod_{j=1}^{N_{\ell}^r} \prod_{k=1}^{N_m^r} |\alpha_{\ell,j} - \alpha_{m,k}|^{\beta Q_{\ell} Q_m} \\
& \qquad \times d\mu^{N_1^r}(\bs \alpha_1) \cdots d\mu^{N_M^r}(\bs \alpha_M).
\end{align*}
Note that each of the integrands appearing in the inner most product is invariant under translation $r + \mf m \mapsto \mf m$, and we find our first major simplification (by appropriate definition, but still!),
\begin{align*}
&  Z(\mathbf t, \mf o, \beta)
   =  \sum_{\mathbf N_1} \frac{t_1^{N_1^{0}} \cdots t_1^{N_1^{q-1}}}{N_1^0! \cdots N_1^{q-1}!} \cdots \sum_{\mathbf N_M} \frac{t_M^{N_M^{0}} \cdots t_M^{N_M^{q-1}}}{N_M^0! \cdots N_M^{q-1}!} \prod_{r=0}^{q-1} Z(\mathbf N^r, \mf m, \beta)
\end{align*}
Now, the sum over $\mathbf N_1, \ldots, \mathbf N_M$ can be thought of as a sum over the rows of all $M \times q$ matrices with positive integer coefficients. We may reindex this sum so that instead of summing over rows, we sum over the columns $\mathbf N^0, \ldots, \mathbf N^{q-1}$. When we do this, we regroup the monomials accordingly
\[
\frac{t_1^{N_1^{0}} \cdots t_1^{N_1^{q-1}}}{N_1^0! \cdots N_1^{q-1}!} \cdots  \frac{t_M^{N_M^{0}} \cdots t_M^{N_M^{q-1}}}{N_M^0! \cdots N_M^{q-1}!} = \frac{t_1^{N_1^{0}} \cdots t_M^{N_M^{0}}}{N_1^0! \cdots N_M^0!} \cdots \frac{t_1^{N_1^{q-1}} \cdots t_M^{N_M^{q-1}}}{N_1^{q-1}! \cdots N_M^{q-1}!} = \frac{\mathbf t^{\mathbf N^0}}{\mathbf N^0!} \cdots \frac{\mathbf t^{\mathbf N^{q-1}}}{\mathbf N^{q-1}!}.
\]
It follows that
\begin{align*}
Z(\mathbf t, \mf o, \beta)
   &=  \sum_{\mathbf N^0} \cdots \sum_{\mathbf N^{q-1}}
\frac{\mathbf t^{\mathbf N^0}}{\mathbf N^0!} \cdots \frac{\mathbf t^{\mathbf N^{q-1}}}{\mathbf N^{q-1}!}
   Z(\mathbf N^0, \mf m, \beta) \cdots Z(\mathbf N^{q-1}, \mf m, \beta) \\
   &=  \sum_{\mathbf N^0}  \frac{\mathbf t^{\mathbf N^0}}{\mathbf N^0!}    Z(\mathbf N^0, \mf m, \beta) \cdots \sum_{\mathbf N^{q-1}} \frac{\mathbf t^{\mathbf N^{q-1}}}{\mathbf N^{q-1}!} Z(\mathbf N^{q-1}, \mf m, \beta) \\
   &= \bigg[\sum_{\mathbf N}  \frac{\mathbf t^{\mathbf N}}{\mathbf N!} Z(\mathbf N, \mf m, \beta) \bigg]^{q} \\
   &= Z(\mathbf t, \mf m, \beta)^q. \qedhere
\end{align*}

\end{proof}

\section{Acknowledgements}

The author has had many conversations about the `$p$-adic Selberg integral' and 'repelling $p$-adic random variables' over the last decade that ultimately lead to the current work. In particular, Oregon colleagues Matt Grimes, Jonathan Wells, Joe Webster and Ben Young have shared their insights on the topic, and a subset of them are working on continutions of the work presented here. Much of the groundwork was laid at a series of small workshops (SQuaRE) hosted by the American Institute of Mathematics where I was joined by Igor Pritsker, Jeff Vaaler and Maxim Yattselev, who encouraged my fascination in the current work. I am indebted to the generosity of these organizations and individuals.

\bibliography{p-electro}

\begin{thebibliography}{10}

\bibitem{MR2482347}
Matthew Baker.
\newblock An introduction to {B}erkovich analytic spaces and non-{A}rchimedean
  potential theory on curves.
\newblock In {\em {$p$}-adic geometry}, volume~45 of {\em Univ. Lecture Ser.},
  pages 123--174. Amer. Math. Soc., Providence, RI, 2008.

\bibitem{MR2229382}
Joe Buhler, Daniel Goldstein, David Moews, and Joel Rosenberg.
\newblock The probability that a random monic {$p$}-adic polynomial splits.
\newblock {\em Experiment. Math.}, 15(1):21--32, 2006.

\bibitem{MR3652559}
B.~Dragovich, A.~Yu. Khrennikov, S.~V. Kozyrev, I.~V. Volovich, and E.~I.
  Zelenov.
\newblock {$p$}-adic mathematical physics: the first 30 years.
\newblock {\em p-Adic Numbers Ultrametric Anal. Appl.}, 9(2):87--121, 2017.

\bibitem{MR0143558}
Freeman~J. Dyson.
\newblock Statistical theory of the energy levels of complex systems. {I-IV}.
\newblock {\em J. Mathematical Phys.}, 3:166--175, 1962.

\bibitem{MR3373827}
Paul Fili and Zachary Miner.
\newblock Equidistribution and the heights of totally real and totally
  {$p$}-adic numbers.
\newblock {\em Acta Arith.}, 170(1):15--25, 2015.

\bibitem{MR3710764}
Paul Fili, Clayton Petsche, and Igor Pritsker.
\newblock Energy integrals and small points for the {A}rakelov height.
\newblock {\em Arch. Math. (Basel)}, 109(5):441--454, 2017.

\bibitem{MR2641363}
P.~J. Forrester.
\newblock {\em Log-gases and random matrices}, volume~34 of {\em London
  Mathematical Society Monographs Series}.
\newblock Princeton University Press, Princeton, NJ, 2010.

\bibitem{MR2434345}
Peter~J. Forrester and S.~Ole Warnaar.
\newblock The importance of the {S}elberg integral.
\newblock {\em Bull. Amer. Math. Soc. (N.S.)}, 45(4):489--534, 2008.

\bibitem{fu2018selberg}
Zenan Fu and Yongchang Zhu.
\newblock Selberg integral over local fields, 2018.

\bibitem{MR1488696}
Fernando~Q. Gouv\^ea.
\newblock {\em {$p$}-adic numbers}.
\newblock Universitext. Springer-Verlag, Berlin, second edition, 1997.
\newblock An introduction.

\bibitem{MR1743467}
Jun-ichi Igusa.
\newblock {\em An introduction to the theory of local zeta functions},
  volume~14 of {\em AMS/IP Studies in Advanced Mathematics}.
\newblock American Mathematical Society, Providence, RI; International Press,
  Cambridge, MA, 2000.

\bibitem{MR2129906}
Madan~Lal Mehta.
\newblock {\em Random matrices}, volume 142 of {\em Pure and Applied
  Mathematics (Amsterdam)}.
\newblock Elsevier/Academic Press, Amsterdam, third edition, 2004.

\bibitem{Ostrowski1916}
Alexander Ostrowski.
\newblock {\"U}ber einige l{\"o}sungen der funktionalgleichung
  $\psi$(x){\textperiodcentered}$\psi$(x)=$\psi$(xy).
\newblock {\em Acta Mathematica}, 41(1):271--284, Dec 1916.

\bibitem{MR0018287}
Atle Selberg.
\newblock Remarks on a multiple integral.
\newblock {\em Norsk Mat. Tidsskr.}, 26:71--78, 1944.

\bibitem{tong}
David Tong.
\newblock Lectures on statistical physics.
\newblock Available at
  \url{https://www.damtp.cam.ac.uk/user/tong/statphys.html} (2020/02/12).

\bibitem{webster2020logcoulomb}
Joe Webster.
\newblock log-coulomb gas with norm-density in $p$-fields, 2020.

\bibitem{MR1608309}
W.~A. Z\'{u}\~{n}iga Galindo.
\newblock Igusa's local zeta functions of semiquasihomogeneous polynomials.
\newblock {\em Trans. Amer. Math. Soc.}, 353(8):3193--3207, 2001.

\end{thebibliography}

\begin{center}
\noindent\rule{4cm}{.5pt}
\vspace{.25cm}

\noindent {\sc \small Christopher D.~Sinclair}\\
{\small Department of Mathematics, University of Oregon, Eugene OR 97403} \\
email: {\tt csinclai@uoregon.edu}
\end{center}

\end{document}